\documentclass[11pt,a4paper,oneside]{amsart}

\usepackage[dvips,colorlinks=true,linkcolor=blue]{hyperref}
\newcommand{\HH}{\mathcal{H}}

\newcommand{\CC}{\mathbb{C}}
\newcommand{\RR}{\mathbb{R}}
\newcommand{\ZZ}{\mathbb{Z}}

\newcommand{\EE}{\mathbb{E}}
\newcommand{\PP}{\mathbb{P}}
\newcommand{\NN}{\mathbb{N}}
\newcommand{\cN}{\mathcal{N}}

\newcommand{\dom}{\mathop{\mathrm{dom}}}
\newcommand{\slim}{\mathop{\mathrm{s{-}lim}}\limits}
\newcommand{\ran}{\mathop{\mathrm{ran}}}

\newcommand{\spec}{\mathop{\mathrm{spec}}}
\newcommand{\diag}{\mathop{\mathrm{diag}}}
\newcommand{\dist}{\mathop{\mathrm{dist}}}

\newcommand{\R}{{\mathbb R}}

\newcommand{\Z}{{\mathbb Z}}

\newcommand{\C}{{\mathbb C}}

\newcommand{\D}{\displaystyle}

\newcommand{\vers}{\operatornamewithlimits{\to}}

\newtheorem{theorem}{Theorem}
\newtheorem{prop}[theorem]{Proposition}

\newtheorem{corol}[theorem]{Corollary}

\theoremstyle{definition}

\newtheorem{rem}[theorem]{\bf Remark}

\begin{document}

\title[Localization on quantum graphs]{Localization on quantum graphs
  with random vertex couplings}

\author{Fr{\'e}d{\'e}ric Klopp}

\address{L.A.G.A., CNRS UMR 7539, Institut Galil{\'e}e, Universit{\'e}
  Paris Nord, 99 avenue Jean-Baptiste Cl{\'e}ment, 93430 Villetaneuse,
  France \& Institut Universitaire de France}
\email{klopp@math.univ-paris13.fr}

\author{Konstantin Pankrashkin} \address{L.A.G.A., CNRS UMR 7539,
  Institut Galil{\'e}e, Universit{\'e} Paris Nord, 99 avenue
  Jean-Baptiste Cl{\'e}ment, 93430 Villetaneuse, France \& Institut
  f{\"u}r Mathematik, Humboldt-Universit{\"a}t zu Berlin, Rudower
  Chaussee 25, 12489 Berlin, Germany} \email{const@math.hu-berlin.de}

\begin{abstract}
  We consider Schr{\"o}dinger operators on a class of periodic quantum
  graphs with randomly distributed Kirchhoff coupling constants at all
  vertices.  Using the technique of self-adjoint extensions we obtain
  conditions for localization on quantum graphs in terms of finite
  volume criteria for some energy-dependent discrete Hamiltonians.
  These conditions hold in the strong disorder limit and at the
  spectral edges.
\end{abstract}

\maketitle



\section*{Introduction}

In the present work we study spectral properties for a special type of
random interactions on quantum graphs, the so-called random Kirchhoff
model. We are going to show that such models can be effectively
treated using well-established methods for the discrete Anderson
model, in particular, with the help of finite volume fractional moment
criteria.

The study of random Schr{\"o}dinger operators on quantum graphs has
become especially active during the last years.  In~\cite{ASW} weakly
disordered tree graphs were studied; it was shown that the absolutely
continuous spectrum is stable in the weak disorder limit.  Random
interaction on radial tree-like graphs were studied in~\cite{HP}; for
the random edge length and random coupling constants it was shown that
the corresponding Schr{\"o}dinger operators exhibit the Anderson
localization at all energies. This generalizes previously known
results on the random necklace graphs~\cite{KSch}.  Schr{\"o}dinger
operators with random potentials on the edges have been studied
in~\cite{EHS} using the multiscale method, where the presence of the
dense pure point spectrum at the bottom of the spectrum was shown.
The authors of~\cite{GLV,HV} have proved the existence of the
integrated density of states and Wegner estimates for periodic quantum
graphs with random interactions (for both random potentials and random
boundary conditions).

Our method consists in a reduction of the spectral problem on quantum
graphs to the study of a family of energy dependent discrete operators
with a random potential. To perform this reduction we use the theory
of self-adjoint extensions, or, more precisely, the machinery of
abstract Weyl functions~\cite{BGP1}.  A reduction of continuous
problems to discrete ones within the localization framework was
exploited in numerous papers on Schr{\"o}dinger operators with random
or quasiperiodic point interactions, see
e.g.~\cite{BMG,DMP,HKK,GM,Kl:93,Kl:95a}, but, as we will see below,
such a correspondence is particularly explicit and efficient for
quantum graphs.

We consider periodic quantum graphs spanned by simple $\ZZ^d$-lattices
with randomly distributed Kirchhoff coupling constants at all vertices
(the precise construction is given in section~\ref{ss1}).  The edges
can carry additional scalar potentials and the quantum graph is not
assumed to be isotropic.  Actually the scheme presented below can be
directly extended to graphs with more complicated combinatorial
properties, but we do not do this to avoid technicalities. The central
points of the paper are theorem~\ref{th-loc}, where a condition of the
Schr{\"o}dinger operator on a quantum graph to have a pure point
spectrum in terms of upper spectral measures is obtained, and
proposition~\ref{prop-measure}, where we provide estimates for the
spectral measures of quantum graphs in terms of associated discrete
operators. These tools reduce the problem to a direct application of
finite volume criteria for discrete Hamiltonians. Using these criteria
we establish localization in the strong disorder regime
(section~\ref{ss-strong}) and localization at the band edges
(section~\ref{ss-lif}) using the Lifshitz asymptotics for the density
of states.

\section{Schr{\"o}dinger operator on a quantum graph}\label{ss1}

\subsection{Construction of Hamiltonians}

For general matters concerning the theory and applications of quantum
graphs, we refer to~\cite{GS,Ku1,Ku2}.

We consider a quantum graph whose set of vertices is identified with
$\ZZ^d$.  By $h_j$, $j=1,\dots,d$, we denote the standard basis
vectors of $\ZZ^d$.

Two vertices $m$, $m'$ are connected by an oriented edge $m\to m'$ iff
$\D |m-m'|:=\sum_{j=1}^d|m_j-m'_j|=1$ and $m_j\le m'_j$ for all
$j=1,\dots,d$; one says that $m$ is the initial vertex and $m'$ is the
terminal vertex.  Hence, each edge $\epsilon$ has the form $m\to
(m+h_j)$ with some $m\in\ZZ^d$ and $j\in\{1,\dots,d\}$; in this case
we will write $\epsilon=(m,j)$.

Fix some $l_j>0$, $j\in\{1,\dots,d\}$, and replace each edge $(m,j)$
by a copy of the segment $[0,l_j]$ in such a way that $0$ is
identified with $m$ and $l_j$ is identified with $m+h_j$.  In this way
we arrive at a certain topological set carrying a natural metric
structure.  We will parameterize the points of the edges by the
distance from the initial vertex.  Point $x$ lying on the edge $(m,j)$
on the distance $t\in[0,l_j)$ from $m$ will be denoted as $x=(m,j,t)$.
There is an ambiguity concerning the coordinates of the vertices, but
this does not influence the constructions below.

The above graph can be embedded into $\RR^d$, if one identifies
$\ZZ^d\ni m\sim p(m):=\sum_{j=1}^d m_j\, l_j h_j\in\RR^d$, $(m,k)\sim
\big[p(m), p(m)+l_k h_k\big]$, but this will not be used.

The quantum state space of the system is
$\D\HH:=\bigoplus_{m\in\ZZ^d}\bigoplus_{j\in\{1,\dots,d\}} \HH_{m,j}$
where $\D\HH_{m,j}=L^2([0,l_j])$, and the elements of $\HH$ will be
denoted by $f=(f_{m,j})$, $f_{m,j}\in\HH_{m,j}$, $m\in\ZZ^d$,
$j=1,\dots,d$, or $f=(f_\epsilon)$, $f_\epsilon\in\HH_\epsilon$,
$\epsilon\in\ZZ^d\times\{1,\dots,d\}$.  In what follows, we denote by
$P_\epsilon= P_{m,j}$ the orthogonal projection from $\HH$ to
$\HH_\epsilon= \HH_{m,j}$, $\epsilon=(m,j)$. We say that a function
$f=(f_{m,j})$ is \emph{concentrated on an edge $(m,j)$} if
$P_{m,j}f=f$, i.e.  if all components of $f$ but $f_{m,j}$ vanish.

Let us describe the Schr{\"o}dinger operator acting in $\HH$. Fix
real-valued potentials $U_j\in L^2([0,l_j])$, $j=1,\dots,d$, and real
constants $\alpha(m)$, $m\in\ZZ^d$. Set $A:=\diag\big(\alpha(m)\big)$;
this is a self-adjoint operator in $l^2(\ZZ^d)$. Denote by $H_A$ the
operator acting as
\begin{subequations}
  \label{eq-sch}
  \begin{equation}
    \label{eq-act}
    (f_{m,j})\mapsto \Big(
    (-\dfrac{d^2}{dt^2}+U_j)f_{m,j}\Big)
  \end{equation}
  on functions $\D(f_{m,j})\in\bigoplus_{m,j} H^2([0,l_j])$ satisfying
  the following boundary conditions:
  \begin{equation}
    \label{eq-cont1}
    f_{m,j}(0)=f_{m-h_k,k}(l_k)=:f(m),\quad j,k=1,\dots,d
  \end{equation}
  (which means the continuity at all vertices) and
  \begin{equation}
    \label{eq:4}
    f'(m)=\alpha(m)f(m),\quad m\in\ZZ^d,
  \end{equation}
\end{subequations}
where
\begin{equation}
  \label{eq:5}
  f'(m):= \sum_{j=1}^d f'_{m,j}(0)-\sum_{j=1}^d f'_{m-h_j,j}(l_j).
\end{equation}
The constants $\alpha(m)$ are usually referred to as \emph{Kirchhoff
  coupling constants}.  The boundary conditions corresponding to zero
Kirchhoff coupling constants are usually called the Kirchhoff boundary
conditions. Non-zero Kirchhoff coupling constants are usually
interpreted as measuring the impurities at the vertices (zero coupling
constants correspond to the \emph{ideal coupling}). Later we will
assume that $\alpha(m)$ are independent identically distributed random
variables, but here we treat first the deterministic case. For
convenience, for $\alpha\in\RR$ we denote by $H_\alpha$ the above
operator $H_A$ with the diagonal $A$, $A=\alpha\,\text{id}$.

Our aim now is to provide a reduction of the spectral problem for
$H_A$ to a family of discrete spectral problems. We will do this using
the machinery of self-adjoint extensions; a self-contained
presentation of this technique in the abstract setting can be found
e.g. in the recent preprint~\cite{BGP1}.

Denote by $S$ the operator acting as \eqref{eq-act} on the functions
$f$ satisfying only the boundary conditions \eqref{eq-cont1}. On the
domain of $S$, one can define linear maps
\begin{equation*}
  f\mapsto \Gamma f:= \big(f(m)\big)_{m\in\ZZ^d}\in l^2(\ZZ^d),\quad
  f\mapsto \Gamma' f:= \big(f'(m)\big)_{m\in\ZZ^d}\in l^2(\ZZ^d)
\end{equation*}
where $f'$ is defined in~\eqref{eq:5}. By the Sobolev embedding
theorems, the maps $\Gamma,\Gamma'$ are well-defined, and the map
$(\Gamma,\Gamma'):\dom S\to l^2(\ZZ^d)\times l^2(\ZZ^d)$ is surjective.
Moreover, by a simple computation, for any $f,g$ in $\dom S$, one has
\begin{equation*}
  \langle f,Sg\rangle-\langle Sf,g\rangle=\langle \Gamma f,\Gamma' g\rangle-
  \langle\Gamma' f,\Gamma g\rangle
\end{equation*}
(see e.g. proposition~1 in~\cite{KP2}).  In the abstract language,
$(\ZZ^d,\Gamma,\Gamma')$ form a \emph{boundary triple} for $S$. This
permits to write a useful formula for the resolvent of $H_A$, which
will play a crucial role below.

First, denote by $H^0$ the restriction of $S$ to $\ker \Gamma$.
Clearly, $H^0$ acts as \eqref{eq-act} on functions $(f_{m,j})$ with
$f_{m,j}\in H^2([0,l_j])$ satisfying the Dirichlet boundary
conditions, $f_{m,j}(0)=f_{m,j}(l_j)=0$ for all $m,j$, and the
spectrum of $H^0$ is just the union of the Dirichlet spectra of the
operators $-\dfrac{d^2}{dt^2}+U_j$ on the segments $[0,l_j]$.

Denote by $\varphi_j$ and $\vartheta_j$ the solutions to $-y''+U_j
y=Ey$ satisfying $\varphi(0;E)=\vartheta'(0;E)=0$ and
$\varphi'(0;E)=\vartheta(0;E)=1$. For short, we denote
$\phi_j(t;E):=\varphi_j(l_j;E)\vartheta_j(t;E)-\vartheta_j(l_j;E)\varphi_j(t;E)$.
Clearly, $\phi_j$ is the solution to the above differential equation
satisfying $\phi_j(l_j;E)=0$ and $-\phi'_j(l_j;E)=1$.

For $E$ outside $\spec H^0$, consider the operator
$\gamma(E):l^2(\ZZ^d)\to \HH$ defined as follows: for $\xi\in
l^2(\ZZ^d)$, $\gamma(E)\xi$ is the unique solution to $(S-E)f=0$ with
$\Gamma f=\xi$.  For each $E$, $\gamma(E)$ is a linear topological
isomorphism between $l^2(\ZZ^d)$ and $\ker(S-E)$. Clearly, in terms of
the functions $\phi_j, \varphi_j,\vartheta_j$ introduced above, one
has
\begin{equation}
  \label{eq-gammae}
  \big(\gamma(E)\xi\big)_{m,j}(t)=
  \dfrac{1}{\varphi_j(l_j;E)}\Big(\xi(m+h_j)\varphi_j(t;E)+\xi(m)\phi_j(t;E)\Big).
\end{equation}
Furthermore, for $E\not\in\sigma(H^0)$, define the operator
$M(E):l^2(\ZZ^d)\to l^2(\ZZ^d)$ by $M(E):=\Gamma'\gamma(E)$. In our
case,
\begin{equation*}
  M(E)\xi(m)=
  \sum_{j=1}^d\dfrac{1}{\varphi_j(l_j;E)} \Big( \xi(m-h_j)+\xi(m+h_j)\Big)
  - \Big( \sum_{j=1}^d
  \dfrac{\vartheta_j(l_j;E)+\varphi'_j(l_j;E)}{\varphi_j(l_j;E)}\Big)\xi(m).
\end{equation*}
We denote for clarity
\begin{equation*}
  a(E):=\sum_{j=1}^d \dfrac{\eta_j(E)}{\varphi_j(l_j;E)}, \quad
  b_j(E):=\dfrac{1}{\varphi_j(l_j;E)},\quad
  \eta_j(E):= \vartheta_j(l_j;E)+\varphi'_j(l_j;E).
\end{equation*}
then
\begin{equation}
  \label{eq-Mab} 
  M(E)\xi(m)=\sum_{j=1}^d b_j(E)\big( \xi(m-h_j)+\xi(m+h_j)\big) -a(E)\xi(m).
\end{equation}
The maps $\gamma$ and $M$ satisfy a number of important properties.
In particular, $\gamma$ and $M$ depend analytically on their argument
(outside $\spec H^0$), and for any admissible real $E$ one has
\begin{equation}
  \label{eq-mgamma}
  \dfrac{dM(E)}{dE}=\gamma^*(E)\gamma(E),        
\end{equation}
and for any non-real $E$ there is $c_E>0$ such that
\begin{equation}
  \label{eq-me}
  \dfrac{\Im M(E)}{\Im E}\ge c_E.
\end{equation}

The resolvents of $H^0$ and $H_A$ are related by the Krein resolvent
formula,
\begin{equation}
  \label{eq-krein}
  (H_A-E)^{-1}=(H^0-E)^{-1}-\gamma(E)\big(M(E)-A\big)^{-1}\gamma^*(\bar E),\quad
  \quad E\notin\spec H^0\cup\spec H_A.
\end{equation}
Moreover, the set $\spec H_A\setminus \spec H^0$ coincides with
$\{E\notin\spec H^0:\, 0\in\spec \big(M(E)-A\big)\}$, and the same
correspondence holds for the eigenvalues with $\gamma(E)$ being an
isomorphism of the corresponding eigensubspaces.

We note that for special quantum graphs one can perform the complete
reduction of the spectral problem to the spectral problem for the
discrete Laplacian on the underlying combinatorial
graph~\cite{BGP1,BGP2,KP}.  In general, the spectrum is rather
complicated and depends on various geometric and arithmetic
parameters, see e.g.~\cite{E95}.

Eq.~\eqref{eq-krein} shows that $(H_{A}-E)^{-1}$ is an integral
operator whose kernel (the Green function) $G_{A}$ has the following
form:
\begin{equation}
  \label{eq-green}
  \begin{split}
    G_{A}\big( (m,j,t), (m',j',t')
    \big)&=\delta_{mm'}\delta_{jj'} G_j(t,t';E)\\
    &\hskip.5cm-\dfrac{1}{\varphi_j(t;E)\varphi_{j'}(t';E)}\,\Big[
    \big(M(E)-A)^{-1}(m,m')\phi_j(t;E) \phi_{j'}(t';E)\\
    &\hskip1cm+\big(M(E)-A)^{-1}(m+h_j,m')\varphi_j(t;E) \phi_{j'}(t';E)\\
    &\hskip1.5cm+\big(M(E)-A)^{-1}(m,m'+h_{j'})\phi_j(t;E) \varphi_{j'}(t';E)\\
    &\hskip2cm+\big(M(E)-A)^{-1}(m+h_j,m'+h_{j'})\varphi_j(t;E)\varphi_{j'}(t';E)
    \Big],
  \end{split}
\end{equation}
where $G_j$ is the Green function for $-d^2/dx^2+U_j$ on
$L^2([0,l_j])$ with the Dirichlet boundary conditions, i.e.
\begin{equation}
  \label{eq-grf1}
  G_j(t,t';E)=
  \begin{cases}
    \dfrac{\varphi_j(t;E)\phi(t';E)}{W_j(E)},& t< t',\\[\bigskipamount]
    \dfrac{\varphi_j(t';E)\phi(t;E)}{W_j(E)}, & t>t',
  \end{cases}\quad
  W_j(E):=\varphi_j(t;E)\phi'_j(t;E)-\varphi'_j(t;E)\phi_j(t;E).
\end{equation}
\subsection{Random Hamiltonians}\label{ss-ran}
On $(\Omega,\PP)$ a probability space, let
$(\alpha_\omega(m))_{m\in\Z^d}$ be a family of independent identically
distributed (i.i.d.) random variables whose common distribution has a
bounded density $\rho$ with support $[\alpha_-,\alpha_+]$.

By a random Hamiltonian acting on the quantum graph, we mean the
family of operators given by Eqs. \eqref{eq-sch} corresponding to the
parameterizing operator $A_\omega:=\{\lambda\alpha_\omega(m)\}$ of
Kirchhoff coupling constants at the vertices, where $\alpha_\omega(m)$
are described above. This family of Hamiltonians will be denoted by
$H_{\lambda,\omega}$ or $H_{A,\omega}$.

For the moment we can set without loss of generality $\lambda=1$ and
denote the Hamiltonians simply by $H_\omega$.

The shifts $\tau_m$, defined by $(\tau_m\omega)_{m'}=\omega_{m+m'}$,
$m,m'\in\ZZ^d$, act as a measure preserving ergodic family on
$\Omega$.  For any $\tau_m$, there exists a unitary map $U_m$ on
$\HH$, $(U_m f)_{m',j'}=f_{m+m',j'}$, $m,m'\in\ZZ^d$,
$j'\in\{1,\dots,d\}$, with $H_{\tau_m\omega}=U^*_m H_\omega U_m$,
which implies the following standard result from the theory of random
operators, the existence of an almost sure spectrum and of almost sure
spectral components (see e.g.~\cite{PF}), i.e. the existence of closed
subsets $\Sigma_\bullet\subset\RR$ and a subset $\Omega'\subset\Omega$
with $\PP(\Omega')=1$ such that $\spec_\bullet
H_{\omega}=\Sigma_\bullet$,
$\bullet\in\{\text{pp},\text{ac},\text{sc}\}$, for any
$\omega\in\Omega'$. Let
$\Sigma=\Sigma_\text{pp}\cup\Sigma_\text{ac}\cup\Sigma_\text{sc}$ be
the almost sure spectrum of $H_\omega$.

By Eq.~\eqref{eq-krein} and the discussion thereafter, for any
$E\notin\spec H^0$ one has the equivalence $E\in\spec H_\omega$ if and
only if $0\in\spec\big(M(E)-A_\omega\big)$.  At the same time,
$M(E)-A_\omega$ is a usual metrically transitive operator in
$l^2(\ZZ^d)$ and hence possesses an almost sure spectrum $\Sigma_M(E)$
which satisfies (see~\cite{PF})
\begin{equation}
  \label{eq-mea}
  \begin{split}
    \Sigma_M(E)&=\spec M(E)-[\alpha_-,\alpha_+]\\
    &= \left[-2\sum_{j=1}^d \big|b_j(E)\big|-a(E)-\alpha_+,
      2\sum_{j=1}^d \big|b_j(E)\big|-a(E)-\alpha_-\right].
  \end{split}
\end{equation}
Hence, the characteristic equation for $E\notin\spec H^0$ to be in the
almost sure spectrum of $H_\omega$ reads
\begin{equation}
  \label{eq-speca}
  \Big(
  2\sum_{j=1}^d \big|b_j(E)\big|-a(E)-\alpha_-
  \Big)\cdot\Big(
  2\sum_{j=1}^d \big|b_j(E)\big|+a(E)+\alpha_+
  \Big)\ge 0.
\end{equation}
So, the spectrum of $H_\omega$ outside the Dirichlet eigenvalues is a
union of bands.

Let us turn to the dependence of $H_{\lambda,\omega}$ on $\lambda$.
The characteristic equation \eqref{eq-speca} for the spectrum becomes
\begin{equation}
  \label{eq-specb}
  \Big(
  2\sum_{j=1}^d \big|b_j(E)\big|-a(E)-\lambda\alpha_-
  \Big)\cdot\Big(
  2\sum_{j=1}^d \big|b_j(E)\big|+a(E)+\lambda\alpha_+
  \Big)\ge 0.
\end{equation}
Let us describe the behavior of the almost sure spectrum as
$\lambda\to+\infty$. Recall the well-known asymptotics~\cite{LS}:
\begin{equation}
  \label{eq-asymp}
  \begin{gathered}
    \eta_j(E)\sim 2\cosh {l_j\sqrt{-E}}, \quad \varphi_j(l_j,E)\sim
    \dfrac{\sinh {l_j\sqrt{-E}}}{\sqrt{-E}}, \quad
    E\to- \infty,\\
    \eta_j(E)\sim 2\cos l_j\sqrt{E}, \quad\varphi_j(l_j,E)\sim
    \dfrac{\sin {l_j\sqrt{E}}}{\sqrt{E}}, \quad E\to+\infty.
  \end{gathered}
\end{equation}
In particular, $b_j(E)=O(e^{-\alpha\sqrt{-E}})$, $\alpha>0$, and
$a(E)\sim 2d\sqrt{-E}$ for $E\to-\infty$.

If $\alpha_-<0<\alpha_+$, then condition~\eqref{eq-specb} can be
satisfied for any $E$ if $\lambda$ is chosen sufficiently large, i.e.
the spectrum tends to cover the whole real axis. The edges of the
spectrum are situated in the domains where the expressions
$2\sum_{j=1}^d \big|b_j(E)\big|\pm a(E)$ are of order $\lambda$; so,
these edges lie in $O(\lambda^{-1})$-neighborhoods of the Dirichlet
eigenvalues and close to $-\infty$.

If $0\in[\alpha_-,\alpha_+]$, then \eqref{eq-specb} will be satisfied
for any $\lambda$ if
\begin{equation*}
  \Big(
  2\sum_{j=1}^d \big|b_j(E)\big|-a(E)\Big)\cdot\Big(
  2\sum_{j=1}^d \big|b_j(E)\big|+a(E)\Big)\ge 0,
\end{equation*}
i.e. the spectrum contains a part which does not depend on $\lambda$;
actually, this part is nothing but the spectrum of the Hamiltonian
$H_0$ corresponding to the zero coupling constants at all vertices
i.e. $\alpha_\omega(m)=$, $\forall m$.

If $\alpha_-$ and $\alpha_+$ are both positive or both negative, for
\eqref{eq-specb} to be satisfied, the expressions $2\sum_{j=1}^d
\big|b_j(E)\big|\pm a(E)$ must be of the same order as $\lambda$, i.e.
must be large.  Therefore, for $\lambda\to+\infty$ the condition
\eqref{eq-speca} can be satisfied only in the following cases:
\begin{itemize}
\item $\varphi(l_j,E)\sim \lambda^{-1}$ for some $j$,
\item $\alpha_+<0$ and $\sqrt{-E}\sim \lambda$.
\end{itemize}
In other words, for $\lambda\to+\infty$ the spectrum on the positive
half-line concentrates in $O(\lambda^{-1})$ neighborhoods of the
Dirichlet eigenvalues.  For $\alpha_+<0$ there is a band going to
infinity on the negative half-line.

Finally, if $\alpha_-<0$ then, there is some spectrum on the negative
half-axis at the energies of order $\sqrt{-E}\sim-\lambda$.
\section{Localization conditions for quantum graphs}
\label{sec:local-cond-quant}
In this section we again set $\lambda=1$ and study the operator
$H_\omega$.  The following spectral characteristics of $H_\omega$ will
be of crucial importance for us.

Let $f,g\in\HH$. Let $\mu^{f,g}$ denote the spectral measure for $H_A$
associated with $H_A$ and $|\mu^{f,g}|$ denote its absolute value.
For any measurable set $F$ and two edges $(m,j)$, $(m',j')$ we set
\begin{equation*}
  \mu^{(m,j),(m',j')}(F):=\sup_{\substack{f=P_{m,j}f,\\g=P_{m',j'}g\\
      \|f\|=\|g\|=1}}|\mu^{f,g}|(F)
\end{equation*}
and call $\mu^{(m,j),(m',j')}$ the \emph{upper spectral measure}
associated with the edges $(m,j)$ and $(m',j')$ and $H_A$.  For the
random Hamiltonian $H_\omega$, the corresponding quantities get an
additional subindex $\omega$. Recall that for $\mu$ a complex valued
regular Borel measure and $F$ a Borel set, one defines
\begin{equation}
  \label{eq:2}
  |\mu|(F)=\sup_{f\in\mathcal{C}_0(\R),\ |f|_\infty\leq1}
  \left|\int_Ff(E)d\mu(E) \right|.
\end{equation}
We provide localization criteria for $H_A$ in terms of the upper
spectral measures; this extends to the quantum graph case the
localization criteria known for discrete Hamiltonians, cf.
theorem~IV.4 and corollary~IV.5 in~\cite{KS}.

\begin{theorem}
  \label{th-loc}
  Let $F\subset\RR$. Assume that, for any $(m,j)$, one has
  \begin{equation}
    \label{eq-mufin}
    \sum_{m'\in\ZZ^d}\sum_{j'=1}^d \mu^{(m,j),(m',j')}(F)<\infty,
  \end{equation}
  then $H_A$ has only pure point spectrum in $F$.
\end{theorem}
\begin{proof}[Proof of theorem~\ref{th-loc}]
  We use the following result from \cite{AG} (theorem on p.~642):
  \begin{prop}\label{thm-AG}
    Let $H$ be a self-adjoint operator in a Hilbert space $\HH$ and
    $F_r$ be a family of orthogonal projections such that
    $\slim_{r\to+\infty} F_r=1$.  Suppose that there exists a family
    $\{S_n\}$ of linear operators, such that each $S_n$ is bounded,
    defined everywhere, and commutes with $H$, and the strong limit
    $\D S:=\mathrm s-\lim_{n\to\infty} S_n$ exists and $\overline{\ran
      S}=\HH$.  Assume additionally that $F_rS_n$ is compact for any
    $r$ and $n$. Then, the invariant subspace $\HH_{pp}$ of $H$
    corresponding to the pure point spectrum admits the following
    description:
    \begin{equation*}
      \HH_{pp}=\big\{
      f\in\HH:\, \lim_{r\to\infty}\sup_{t\in\RR}(1-F_r) e^{itH} f=0
      \big\}.
    \end{equation*}
  \end{prop}
  and the the technical result
  \begin{prop}
    \label{prop-comp}
    Let $\Lambda$ be a subset of $\ZZ^d$. Denote by $P_\Lambda$ the
    orthogonal projection from $\HH$ to the span of the functions
    $(f_{m,j})$ with $f_{m,j}= 0$ for $m\notin\Lambda$.  For any
    finite $\Lambda$ and any $E\notin\spec H_A$, the operator
    $T:=P_\Lambda (H_A-E)^{-1}$ is Hilbert-Schmidt, hence compact.
  \end{prop}
  \noindent that we prove in the appendix~\ref{sec:appendix}.

  Denote by $P_F$ denote the spectral projection onto $F$
  corresponding to $H_{A}$.  It is sufficient to show that $P_F f$
  belongs to the invariant space of $H_A$ associated with the point
  spectrum for any $f\in\HH$. Clearly, it suffices to consider only
  functions $f$ concentrated on a single edge.

  Let us use proposition~\ref{thm-AG}. Take $S= S_n= (H_A-i)^{-1}$.
  As $F_r$ we take the orthogonal projections from $\HH$ to the
  functions $(f_{m,j})$ with $f_{m,j}=0$ for $|m|>r$.  Clearly, $S$ is
  bounded, commutes with $H_{A}$, $\ran S=\dom H_{A}$ is dense in
  $\HH$, $F_r S$ is compact for any $r$ due to
  proposition~\ref{prop-comp}, and $F_r$ strongly converge to the
  identity operator. Hence, the assumptions of
  proposition~\ref{thm-AG} are satisfied.

  Take any $f$ with $f=P_{m,j}f$. Clearly, in our setting,
  \begin{equation*}
    \begin{split}
    \sup_{t\in\RR}\|(1-F_r)e^{-itH_A} P_F f)\|^2
    &=\sup_{t\in\RR}\sum_{|m'|>r}\sum_{j'=1}^d
    \big\|\big(e^{-itH_A} P_F f\big)_{m',j'}\big\|^2\\
    &=\sup_{t\in\RR}\sum_{|m'|>r}\sum_{j'=1}^d
    \langle e^{-itH_A} P_F f, P_{m',j'}e^{-itH_A} P_F f\rangle\\
    &\le \sum_{|m'|>r}\sum_{j'=1}^d \sup_{t,s\in\RR} \Big|\big\langle
    e^{-isH_A} P_F f, P_{m',j'}e^{-itH_A} P_F f \big\rangle \Big|.      
    \end{split}
  \end{equation*}
  Due to the definition of the absolute value of a measure one has
  \begin{equation*}
    \sup_{s\in\RR} \Big|\big\langle
    e^{-isH_A} P_F f, P_{m',j'} e^{-itH_A} P_F f \big\rangle
    \Big|
    \le
    \big|\mu^{f, P_{m',j'}e^{-itH_A} P_F f}\big|(F).
  \end{equation*}
  Using the definition of $\mu^{(m,j),(m',j')}(F)$ one obtains
  \begin{equation*}
    \sup_{t\in\RR}\big|\mu^{f, P_{m',j'}e^{-itH_A} P_F f}\big|(F) \le
    \mu^{(m,j),(m',j')}(F)\|P_F f\| \,\|f\|
    \le \mu^{(m,j),(m',j')}(F)\|f\|^2.
  \end{equation*}
  Finally, we obtain
  \begin{equation*}
    \sup_{t\in\RR}\big\|(1-F_r)e^{-itH_A} P_F f)\big\|^2\le
    \|f\|^2\,\sum_{|m'|>r}\sum_{j'=1}^d \mu^{(m,j),(m',j')}(F),
  \end{equation*}
  and by \eqref{eq-mufin},
  $\lim\limits_{r\to+\infty}\sup_{t\in\RR}\big\|(1-F_r)e^{-itH_A} P_F
  f\big\|^2=0$.
\end{proof}

Theorem~\ref{th-loc} admits a direct application to the random
Hamiltonians $H_\omega$.
\begin{corol}
  \label{corol1}
  Let $F\subset\RR$. Assume that, for any edge $(m,j)$, one has
  \begin{equation}
    \label{eq-mufin2}
    \EE\Big(\sum_{m'\in\ZZ^d}\sum_{j'=1}^d
    \mu_\omega^{(m,j),(m',j')}(F)\Big)<\infty,
  \end{equation}
  then $H_{\omega}$ has only pure point spectrum in $F$ almost surely.
\end{corol}

\begin{proof}
  Eq.~\eqref{eq-mufin2} says, in particular, that for any $(m,j)$
  there exists $\Omega_{m,j}\subset\Omega$ with $\PP(\Omega_{m,j})=1$
  such that, for $\omega\in\Omega_{m,j}$, $\D\sum_{m'\in\ZZ^d}
  \sum_{j'=1}^d\mu_\omega^{(m,j),(m',j')}(F)<\infty$. \\
  Denote $\D\Omega':=\bigcap_{m,j}\Omega_{m,j}$; as the set of all
  $(m,j)$ is countable, $\PP(\Omega')=1$. Clearly,
  $\D\sum_{m'\in\ZZ^d}\sum_{j'=1}^d
  \mu_\omega^{(m,j),(m',j')}(F)<\infty$ for all $(m,j)$ and all
  $\omega\in\Omega'$, and the spectrum of $H_\omega$ in $F$ is pure
  point for any $\omega\in\Omega'$ by theorem~\ref{th-loc}.
\end{proof}

In the next result, we show that assumption~\eqref{eq-mufin2} is a
consequence of a finite volume criteria \emph{{\`a} la}~\cite{ASFH} on the
discrete Hamiltonians defined in section~\ref{ss1}. The finite volume
criteria is expressed in terms of finite volume approximations of our
operators that we define first.

Let $\Lambda$ be a subset of $\ZZ^d$. Denote by $H^\Lambda_A$ the
operator acting by the same rule \eqref{eq-act} on functions $f$
satisfying the boundary conditions $f'(m)=\alpha(m)f(m)$ for
$m\in\Lambda$ and the Dirichlet boundary conditions $f(m)=0$ for
$m\notin\Lambda$.  In other words, the functions from the domain
$H^\Lambda_A$ satisfy the same boundary conditions as for $H_A$ at the
vertices lying in $\Lambda$ and those as for $H^0$ at the vertices
outside $\Lambda$.  One can relate the operators of $H_A^\Lambda$ and
$H^0$ by a formula similar to~\eqref{eq-krein} using e.g. the
construction of~\cite{pos}.

Namely, consider $l^2(\Lambda)$ as a subset of $l^2(\ZZ^d)$ and denote
by $\Pi_\Lambda$ the orthogonal projection from $l^2(\ZZ^d)$ to
$l^2(\Lambda)$.  Denote also $M_\Lambda(E):=P_\Lambda M(E)
\Pi_{\Lambda}$, $A_\Lambda:=\Pi_\Lambda A \Pi_{\Lambda}$; these two
operators are to be considered as acting in $l^2(\Lambda)$, and
$\gamma_\Lambda(E) = \gamma(E)\Pi_\Lambda$, then, for $E\notin\spec
H^0\cup\spec H^\Lambda_A$, the following resolvent formula holds:
\begin{equation}
  \label{eq-krein2}
  (H^\Lambda_A-E)^{-1}=(H^0-E)^{-1}-\gamma_\Lambda(E)
  \big(M_\Lambda(E)-A_\Lambda\big)^{-1}\gamma^*_\Lambda(\bar E).
\end{equation}
As previously, for any $E\notin\spec H^0$ one has
$\ker(H^\Lambda_A-E)=\gamma_\Lambda(E)\ker\big(M_\Lambda(E)-
A_\Lambda\big)$.

In the appendix~\ref{sec:appendix}, we prove the following auxiliary
result
\begin{prop}
  \label{prop-conv}
  Denote $\Lambda_N:=\{m\in\ZZ^d:\,\max_j |m_j|\le N\}$, $N\in\NN$,
  then the operators $H^{\Lambda_N}_A$ converge to $H_A$ in the strong
  resolvent sense as $N\to\infty$.
\end{prop}

that will be used in the proof of our localization criterion.

\begin{prop}
  \label{prop-measure}
  Let $F\subset\RR$ be a segment containing no Dirichlet eigenvalues.
  Assume that there exists $A, a>0$ and $s\in(0,1)$ such that
  \begin{equation}
    \label{eq-Efin}
    \EE\left|\left(M_{\Lambda}(E)-A_{\Lambda,\omega}\right)^{-1}(m,m')
    \right|^s\le A e^{-a|m-m'|} 
  \end{equation}
  for all finite $\Lambda\subset\ZZ^d$ and all $E\in F$.  Then, there
  exist $B,c>0$ such that for any two edges $(m,j)$ and $(m',j')$ one
  has
  \begin{equation}
    \label{eq-PF}
    \EE\big(\mu^{(m,j),(m',j')}_\omega(F)\big)\le B e^{-c|m-m'|}.      
  \end{equation}
\end{prop}

\begin{rem}\label{rem3} By theorem~\ref{th-loc}, the result of
  proposition \ref{prop-measure} clearly implies that, under the
  assumptions of proposition~\ref{prop-measure}, the spectrum is
  almost surely pure point in $F$. By the results of~\cite{ASFH}, in
  particular, theorem 4.1 therein, the assumption of
  proposition~\ref{prop-measure} also implies that, for $E\in F$, the
  spectrum of $M(E)-A_\omega$ is localized in an open interval
  containing $0$. Hence, using the remark following Krein's resolvent
  formula, equation~\eqref{eq-krein}, for $E$ in the spectrum of
  $H_\omega$ and not an eigenvalue of $H_0$ (i.e. not a Dirichlet
  eigenvalue), $0$ is an eigenvalue for $M(E)-A_\omega$. It is
  associated to an eigenfunction, say $\xi$, that is exponentially
  localized in $\Z^d$.  The corresponding eigenfunction for $H_\omega$
  at energy $E$, say, $\varphi$ is then given by
  $\varphi=\gamma(E)\xi$.  By~\eqref{eq-gammae}, $\varphi$ is also
  exponentially localized in the sense that there exists $C>0$ such
  that
\begin{equation*}
  \sup_{1\leq j\leq d}\|\varphi\|_{\mathcal{H}_{m,j}}\leq Ce^{-|m|/C}.
\end{equation*}
Moreover, as in the appendix A of~\cite{ASFH}, by~\eqref{eq:2},
proposition \ref{prop-measure} implies dynamical localization bounds for the operator
$H_A$ in the following sense
\begin{equation}
  \label{eq:3}
  \EE\left(\sup_{\substack{f=P_{m,j}f,\\g=P_{m',j'}g\\
        \|f\|=\|g\|=1}}|\langle
    f,e^{itH_\omega}\mathbf{1}_F(H_\omega)g\rangle|\right)\leq
  Ce^{-|m-m'|/C}.
\end{equation}
\end{rem}

\begin{proof}[\bf Proof of proposition \ref{prop-measure}]
  In view of proposition~\ref{prop-conv}, $H^{\Lambda_n}_{A,\omega}$
  converges to $H_{A,\omega}$ in the strong resolvent sense for a
  suitable choice of finite $\Lambda_n\subset\ZZ^d$ and any $\omega$.
  This implies the weak convergence
  $\mu^{f,g}_{\Lambda,\omega}\to\mu^{f,g}_\omega$ for any
  $f,g,\omega$.  Consequently, by the Fatou lemma, for any $F$ one has
  $\EE\big(\mu^{(m,j),(m',j')}_{\omega} (F)\big)\le \lim\inf
  \EE\big(\mu^{(m,j),(m',j')}_{\Lambda,\omega}(F)\big)$.  In other
  words, it is sufficient to show the existence of positive $B$ and
  $c$ such that for any $(m,j)$ and $(m',j')$ the estimate
  $\EE\big(\mu^{(m,j),(m',j')}_{\Lambda,\omega}(F)\big)\le
  Be^{-c|m-m'|}$ holds for sufficiently large $\Lambda$.  In proving
  this estimate, we follow essentially the steps of~\cite[theorem
  A.1]{ASFH} or \cite[lemma 3.1]{Az}.

  Pick two edges $(m,j)$ and $(m',j')$ and consider $\Lambda\subset
  \ZZ^d$ containing $m$ and $m'$ and all vetrices $n$ with $|n-m'|\le
  2$.

  Denote $\Hat A_\omega:=A_\omega+(\Hat v-\alpha(m')) \Pi_{m'}$, where
  $\Pi_{m'}$ is the projection onto $\delta_{m'}$ and $\Hat v$ is
  distributed identically to $\alpha(m')$, and consider the modified
  Hamiltonian $H_{\Hat A,\omega}$. Note that under our assumptions
  $\EE \big(|\Hat v|^\delta\big)<\infty$ for any $\delta>0$. For
  almost every $\hat{v}$, if $0$ is an eigenvalue of
  $M_\Lambda(E)-A_{\Lambda,\omega}$, then $M_\Lambda(E)-\Hat
  A_{\Lambda,\omega}$ is invertible. Consider also the operators
  $\Tilde A_\omega:=A_\omega+(\Tilde v-\alpha(m'+h_{j'}))
  \Pi_{m'+h_{j'}}$ with $\Tilde v$ distributed identically to
  $\alpha(m'+h_{j'})$, to which the previous observations apply as
  well.

  We note that the spectrum of $H^{\Lambda}_{A,\omega}$ outside the
  Dirichlet eigenvalues is discrete.  Almost surely, each eigenvalue
  of $M_\Lambda(E)-A_\Lambda$ is simple.  One has
  \begin{equation}
    \label{eq-mu1}
    \mu^{f,g}_{\Lambda,\omega}(F)=\sum_{E_k\in\spec
      H^\Lambda_\omega\cap F}
    \dfrac{\langle f,\gamma_\Lambda(E_k)\xi_k\rangle\,\langle
      \gamma_\Lambda(E_k)\xi_k,g\rangle}{\|\gamma_\Lambda(E_k)\xi_k\|^2},
  \end{equation}
  where $E_k$ and $\xi_k$ satisfy
  $(M(E_k)-A_{\Lambda,\omega})\xi_k=0$, $\xi_k\ne 0$.

  Let $E\notin\spec H^0$. In the space $\HH_{m',j'}= L^2[0,l_{j'}]$
  consider the subspace $L(E)$ spanned by the linearly independent
  functions $\varphi_{j'}(E):=\varphi_{j'}(\cdot,E)$ and
  $\phi_{j'}(E):=\phi_{j'}(\cdot,E)$. Denote by $P(E)$ the orthogonal
  projection from $\HH_{m',j'}$ to $L(E)$. Any function $h\in L(E)$
  can be uniquely represented in the form $h=\Hat h+\Tilde h$ with
  $\Hat h,\Tilde h\in L(E)$, $\Hat h\perp \varphi_{j'}(E)$, $\Tilde
  h\perp \phi_{j'}(E)$.  Denote the corresponding projections $L(E)\ni
  h \mapsto \Hat h\in L(E)$ and $L(E)\ni h \mapsto \Tilde h\in L(E)$
  by $\Hat P(E)$ and $\Tilde P(E)$, respectively. In view of the
  analytic dependence of $\varphi_{j'}(E)$ and $\phi_{j'}(E)$, the
  norms of the operators $\Hat P(E) P(E)$ and $\Tilde P(E) P(E)$ are
  uniformly bounded,
  \begin{equation}
    \label{eq-pepe}
    \|\Hat P(E) P(E)\|+\|\Tilde P(E) P(E)\|\le p,\quad p>0, \quad E\in F.
  \end{equation}

  From now on we assume that $f=P_{m,j}f$ and $g=P_{m',j'}g$. Having
  in mind the explicit expression for $\gamma(E)$
  (see~\eqref{eq-gammae}), we compute
  \begin{equation}
    \label{eq-gste}
    [\gamma_\Lambda^*(E)g](m)=\sum_{s=1}^d
    \dfrac{1}{\overline{\varphi_s(l_s;E)}}\big(\langle
    \varphi_s(E),g_{m-h_s,s}\rangle+\langle
    \phi_j(E),g_{m,s}\rangle\big),\quad m\in\Lambda,
  \end{equation}
  and one concludes that, for any $E\in F$, one has
  $\gamma^*(E)g=\gamma^*(E)\Hat P(E) P(E) g +\gamma^*(E)\Tilde P(E)
  P(E) g$, which permits us to rewrite \eqref{eq-mu1} in the form
  \begin{multline}
    \label{eq-mu1a}
    \mu^{f,g}_{\Lambda,\omega}(F)= \sum_{E_k\in\spec
      H^\Lambda_\omega\cap F} \dfrac{\langle
      f,\gamma_\Lambda(E_k)\xi_k\rangle\,\langle
      \gamma_\Lambda(E_k)\xi_k,\Hat P(E_k) P(E_k) g
      \rangle}{\|\gamma_\Lambda(E_k)\xi_k\|^2}\\+ \sum_{E_k\in\spec
      H^\Lambda_\omega\cap F} \dfrac{\langle
      f,\gamma_\Lambda(E_k)\xi_k\rangle\,\langle
      \gamma_\Lambda(E_k)\xi_k,\Tilde P(E_k) P(E_k) g
      \rangle}{\|\gamma_\Lambda(E_k)\xi_k\|^2}.
  \end{multline}

  Denote
  \begin{equation*}
    \Hat
    \varphi_E:=\dfrac{\big(M_\Lambda(E)-A_{\Lambda,\omega}\big)^{-1}
      \delta_{m'}}{\langle\delta_{m'},
      \big(M_\Lambda(E)-A_{\Lambda,\omega}\big)^{-1}\delta_{m'}\rangle}
    =\dfrac{\big(M_\Lambda(E)-\Hat A_{\Lambda,\omega}\big)^{-1}
      \delta_{m'}}{\langle\delta_{m'},
      \big(M_\Lambda(E)-\Hat A_{\Lambda,\omega}\big)^{-1}\delta_{m'}\rangle}.
  \end{equation*}

  Assume that $\xi$ is an eigenvector of
  $M_\Lambda(E)-A_{\Lambda,\omega}$ corresponding to the eigenvalue
  $0$. Then $0=(M_\Lambda(E)-A_{\Lambda,\omega})\xi=(M_\Lambda(E)-\Hat
  A_{\Lambda,\omega})\xi +\big(\Hat v-\alpha(m')\big)\Pi_{m'}\xi$.
  Almost surely the matrix $M_\Lambda(E)-\Hat A_{\Lambda,\omega}$ is
  invertible and $\Pi_{m'}\xi\ne 0$ (otherwise $\xi$ would be an
  eigenvector of $M_\Lambda(E)-\Hat A_{\Lambda,\omega}$).  Hence,
  $\xi=\big(\alpha(m')-\Hat v\big)\langle
  \delta_{m'},\xi\rangle(M_\Lambda(E)-\Hat A_{\Lambda,\omega})^{-1}
  \delta_{m'}$. This means, that $\xi= C\Hat \varphi_E$ with a
  suitable constant $C$.

  By a direct calculation,
  $\big(M_\Lambda(E)-A_{\Lambda,\omega}\big)\Hat
  \varphi_E=(\alpha(m')-\Hat v-\Hat\Gamma(E))\delta_{m'}$, where
  \begin{equation*}
    \Hat\Gamma(E)=-\dfrac{1}{\langle\delta_{m'},
      \big(M_\Lambda(E)-\Hat A_{\Lambda,\omega}\big)^{-1}\delta_{m'}\rangle}.
  \end{equation*}
  Hence, the spectrum of $H^{\Lambda_n}_{A,\omega}$ in $F$ is
  determined by the condition $\alpha(m')-\Tilde v=\Hat\Gamma(E)$, and
  $\Hat \varphi_E$ are the corresponding (non-normalized)
  eigenfunctions. Clearly, one has always
  \begin{equation}
    \label{eq-nf1}
    \langle \delta_{m'},\Hat \varphi_E\rangle=1.
  \end{equation}
  Using these observations one can write almost surely
  \begin{equation}
    \label{eq-mu2a}
    \dfrac{\langle f,\gamma_\Lambda(E_k)\xi_k\rangle\,\langle
      \gamma_\Lambda(E_k)\xi_k, g
      \rangle}{\|\gamma_\Lambda(E_k)\xi_k\|^2}=
    \dfrac{\langle f,\gamma_\Lambda(E_k)\Hat \varphi_{E_k}\rangle
      \langle \gamma_\Lambda(E_k)\Hat
      \varphi_{E_k},g\rangle}{\|\gamma_\Lambda(E_k)\Hat
      \varphi_{E_k}\|^2}. 
  \end{equation}
  Exactly in the same way one shows that the spectrum can be
  determined from the condition $\alpha(m'+h_{j'})-\Tilde
  v=\Tilde\Gamma(E)$ with
  \begin{equation*}
    \Tilde\Gamma(E)=-\dfrac{1}{\langle\delta_{m'+h_{j'}},
      \big(M_\Lambda(E)-\Tilde A_{\Lambda,\omega}\big)^{-1}
      \delta_{m'+h_{j'}}\rangle} 
  \end{equation*}
  and that
  \begin{equation}
    \label{eq-mu2b}
    \dfrac{\langle f,\gamma_\Lambda(E_k)\xi_k\rangle\,\langle
      \gamma_\Lambda(E_k)\xi_k, g
      \rangle}{\|\gamma_\Lambda(E_k)\xi_k\|^2}=
    \dfrac{\langle f,\gamma_\Lambda(E_k)\Tilde\varphi_{E_k}\rangle
      \langle
      \gamma_\Lambda(E_k)\Tilde\varphi_{E_k},g\rangle}
    {\|\gamma_\Lambda(E_k)\Tilde\varphi_{E_k}\|^2}, 
  \end{equation}
  where
  \begin{equation*}
    \Tilde\varphi_E:=\dfrac{\big(M_\Lambda(E)-A_{\Lambda,\omega}\big)^{-1}
      \delta_{m'+h_{j'}}}{\langle\delta_{m'+h_{j'}},
      \big(M_\Lambda(E)-A_{\Lambda,\omega}\big)^{-1}\delta_{m'+h_{j'}}\rangle}
    =\dfrac{\big(M_\Lambda(E)-\Tilde
      A_{\Lambda,\omega}\big)^{-1}\delta_{m'+h_{j'}}}
    {\langle\delta_{m'},\big(M_\Lambda(E)-\Tilde
      A_{\Lambda,\omega}\big)^{-1}\delta_{m'+h_{j'}}\rangle}.
  \end{equation*}
  and obviously
  \begin{equation}
    \label{eq-nf2}
    \langle \delta_{m'+h_{j'}},\Tilde \varphi_E\rangle=1.
  \end{equation}
  Combining the representions \eqref{eq-mu1} and \eqref{eq-mu1a} for
  the spectral measures with the identitites \eqref{eq-mu2a} and
  \eqref{eq-mu2b} one obtain the following:
  \begin{equation}
    \label{eq-mrepr1}
    \begin{split}
      \mu^{f,g}_{\Lambda,\omega}(dE)&= \dfrac{\langle
      f,\gamma_\Lambda(E)\Hat \varphi_{E}\rangle \langle
      \gamma_\Lambda(E)\Hat
      \varphi_{E},g\rangle}{\|\gamma_\Lambda(E)\Hat \varphi_{E}\|^2}
    \cdot \Big( \sum_k \delta(E-E_k) \Big)dE\\
    &=\dfrac{\langle f,\gamma_\Lambda(E)\Tilde \varphi_{E}\rangle
      \langle \gamma_\Lambda(E)\Tilde
      \varphi_{E},g\rangle}{\|\gamma_\Lambda(E)\Tilde \varphi_{E}\|^2}
    \cdot \Big( \sum_k \delta(E-E_k) \Big)dE\\
    &=\dfrac{\langle f,\gamma_\Lambda(E)\Hat \varphi_{E}\rangle \langle
      \gamma_\Lambda(E)\Hat \varphi_{E},\Hat
      P(E)P(E)g\rangle}{\|\gamma_\Lambda(E)\Hat \varphi_{E}\|^2}
    \cdot \Big( \sum_k \delta(E-E_k) \Big)dE\\
    &\hskip1cm+\dfrac{\langle f,\gamma_\Lambda(E)\Tilde \varphi_{E}\rangle
      \langle \gamma_\Lambda(E)\Tilde \varphi_{E},\Tilde
      P(E)P(E)g\rangle}{\|\gamma_\Lambda(E)\Tilde \varphi_{E}\|^2}
    \cdot \Big( \sum_k \delta(E-E_k) \Big)dE.
    \end{split}
  \end{equation}
  Now, note that
  \begin{equation*}
    \sum \delta(E-E_k)=-\delta(\alpha(m')-\Hat
    v-\Hat\Gamma(E))\Hat\Gamma'(E)= -\delta(\alpha(m'+h_{j'})-\Tilde
    v-\Tilde\Gamma(E))\Tilde\Gamma'(E)
  \end{equation*}
  and that, using~\eqref{eq-Mab} and~\eqref{eq-mgamma}, one obtains,
  \begin{equation*}
    \Hat\Gamma'(E)=-\Hat\Gamma^2(E)
    \langle\delta_{m'},(M_\Lambda(E)-\Hat
    A_{\Lambda,\omega})^{-1}M'_\Lambda(E)(M_\Lambda(E)-\Hat
    A_{\Lambda,\omega})^{-1}\delta_{m'}\rangle=
    -\|\gamma_\Lambda(E)\Hat\varphi_E\|^2.
  \end{equation*}
  and
  \begin{equation*}
    \Tilde\Gamma'(E)= -\Tilde\Gamma^2(E)
    \langle\delta_{m'+h_{j'}},(M_\Lambda(E)-\Tilde
    A_{\Lambda,\omega})^{-1}M'_\Lambda(E)(M_\Lambda(E)-\Tilde
    A_{\Lambda,\omega})^{-1}\delta_{m'+h_{j'}}\rangle=
    -\|\gamma_\Lambda(E)\Tilde\varphi_E\|^2. 
  \end{equation*}
  This allows one to rewrite \eqref{eq-mrepr1} as
  \begin{equation}
    \label{eq-mrepr2}
    \begin{split}
      \mu^{f,g}_{\Lambda,\omega}(dE)&= \delta(\alpha(m')-\Hat
      v-\Hat\Gamma(E)) \langle f,\gamma_\Lambda(E)\Hat
      \varphi_{E}\rangle \langle \gamma_\Lambda(E)\Hat
      \varphi_{E},g\rangle dE\\
      &= \delta(\alpha(m'+h_{j'})-\Tilde v-\Tilde\Gamma(E)) \langle
      f,\gamma_\Lambda(E)\Tilde \varphi_{E}\rangle \langle
      \gamma_\Lambda(E)\Tilde \varphi_{E},g\rangle dE\\
      &= \delta(\alpha(m')-\Hat v-\Hat\Gamma(E))\langle f,\gamma_\Lambda(E)\Hat \varphi_{E}\rangle \langle
      \gamma_\Lambda(E)\Hat \varphi_{E}, \Hat P(E) P(E)g\rangle dE\\
      &\hskip1cm+ \delta(\alpha(m'+h_{j'})-\Tilde
      v-\Tilde\Gamma(E)) \langle
      f,\gamma_\Lambda(E)\Tilde \varphi_{E}\rangle \langle
      \gamma_\Lambda(E)\Tilde \varphi_{E},\Tilde P(E) P(E) g\rangle
      dE.
    \end{split}
  \end{equation}

  According to the general properties of spectral measures, one always
  has $\mu^{f,g}_{\Lambda,\omega}(dE)=\Psi^{f,g}(E)
  \mu^{f,f}_{\Lambda,\omega}(dE)$, where $\Psi$ is a measurable
  function satisfying
  \begin{equation*}
    \int_\RR |\Psi^{f,g}(E)|^2 \mu^{f,f}_{\Lambda,\omega}(dE)\le
    \|g\|^2 \|f\|^2.
  \end{equation*}
  In our case, the first two equalities in \eqref{eq-mrepr2} imply
  \begin{gather}
    \label{eq-est1}
    \int_\RR \big|\langle \gamma_\Lambda(E)\Hat \varphi_E,
    h\rangle\big|^2 \delta\big(\alpha(m')-\Hat v-\Hat\Gamma(E)\big)dE
    \le \|h\|^2\\
    \intertext{and}
    \label{eq-est2}
    \int_\RR \big|\langle \gamma_\Lambda(E)\Tilde \varphi_E,
    h\rangle\big|^2 \delta(\alpha(m'+h_{j'})-\Tilde v-\Tilde\Gamma(E))
    dE \le \|h\|^2
  \end{gather}
  for any $h$.

  Now we use the third respresentation in \eqref{eq-mrepr2} for the
  spectral measure to estimate the upper spectral measure for the
  edges $(m,j)$ and $(m',j')$.  Clearly,
  \begin{equation*}
    \begin{split}
      |\mu^{f,g}_{\Lambda,\omega}|(F)&= \int_F \big| \langle
      f,\gamma_\Lambda(E)\Hat \varphi_{E}\rangle \langle
      \gamma_\Lambda(E)\Hat \varphi_{E}, \Hat P(E) P(E)g\rangle\big|
      \delta(\alpha(m')-\Hat v-\Hat\Gamma(E)) dE\\
      &+ \int_F \big| \langle f,\gamma_\Lambda(E)\Tilde
      \varphi_{E}\rangle \langle \gamma_\Lambda(E)\Tilde
      \varphi_{E},\Tilde P(E) P(E) g\rangle\big|
      \delta(\alpha(m'+h_{j'})-\Tilde v-\Tilde\Gamma(E)) dE.
    \end{split}
  \end{equation*}

  The construction of the operators $\Hat P(E)P(E)$ and $\Hat P(E)P(E)$
  implies that $\Pi_{m'}\gamma^*(E)_\Lambda \Hat P(E)P(E)=\gamma^*(E)
  \Hat P(E)P(E)$ and $\Pi_{m'+h_{j'}}\gamma^*_\Lambda(E) \Tilde
  P(E)P(E)=\gamma^*(E)\Tilde P(E)P(E)$.  Together with the
  normalization conditions \eqref{eq-nf1} and \eqref{eq-nf2}, for any
  $g$, this implies
  \begin{equation}
    \label{eq-gp}
    \begin{gathered}
      \Big|\big\langle \gamma_\Lambda(E)\Hat \varphi_E, \Hat P(E)P(E)
      g \big\rangle \Big|= \big\| \gamma^*_\Lambda(E)\Hat P(E)P(E) g
      \big\|,\\
      \Big|\big\langle \gamma_\Lambda(E)\Tilde \varphi_E, \Tilde
      P(E)P(E) g \big\rangle \Big|= \big\| \gamma^*_\Lambda(E)\Tilde
      P(E)P(E) g \big\|.
    \end{gathered}
  \end{equation}
  Now, we estimate
  \begin{equation}
    \label{eq-mup1}
    \begin{split}
      &\EE\left(\sup_{\|f\|=\|g\|=1} |\mu^{f,g}_{\Lambda,\omega}|(F)\right)\\
      &\le \EE\left(\sup_{\|f\|=\|g\|=1} \int_F \big| \langle
        f,\gamma_\Lambda(E)\Hat \varphi_{E}\rangle \langle
        \gamma_\Lambda(E)\Hat \varphi_{E}, \Hat P(E)
        P(E)g\rangle\big|
        \delta(\alpha(m')-\Hat v-\Hat\Gamma(E)) dE\right)\\
      &+\EE\left(\sup_{\|f\|=\|g\|=1} \int_F \big| \langle
        f,\gamma_\Lambda(E)\Tilde \varphi_{E}\rangle \langle
        \gamma_\Lambda(E)\Tilde \varphi_{E},\Tilde P(E) P(E)
        g\rangle\big|
        \delta(\alpha(m'+h_{j'})-\Tilde v-\Tilde\Gamma(E)) dE\right).
    \end{split}
  \end{equation}
  Using \eqref{eq-pepe} and \eqref{eq-gp}, one gets
  \begin{equation}
    \label{eq-mup2}
    \begin{split}
      &\EE\left(\sup_{\|f\|=\|g\|=1} \int_F \big| \langle
      f,\gamma_\Lambda(E)\Hat \varphi_{E}\rangle \langle
      \gamma_\Lambda(E)\Hat \varphi_{E}, \Hat P(E) P(E)g\rangle\big|
    \delta(\alpha(m')-\Hat v-\Hat\Gamma(E)) dE\right)\\
    &\leq p\, G\, \EE\left(\sup_{\|f\|=1} \int_F \left| \langle
      f,\gamma_\Lambda(E)\Hat \varphi_{E}\rangle\right|
      \delta(\alpha(m')-\Hat v-\Hat\Gamma(E)) dE\right),
    \end{split}
  \end{equation}
  where $\D G:=\sup_{E\in F} \|\gamma^*_\Lambda(E)\|<\infty$.  Using
  the H{\"o}lder inequality and \eqref{eq-est1}, one obtains
  \begin{equation*}
    \begin{split}
      &\EE\left(\sup_{\|f\|=1}\int_F\left|\langle
          f,\gamma_\Lambda(E)\Hat
          \varphi_E\rangle\right|\delta\big(\alpha(m')-\Hat
        v-\Hat\Gamma(E)\big)dE\right)\\
      &\leq \left[ \EE\left(|\alpha(m')-\Hat v|^\alpha\sup_{\|f\|=1}
          \int_F \left|\langle f,\gamma_\Lambda(E)(M_\Lambda(E)-\Hat
            A_\Lambda)^{-1}\delta_{m'}\rangle\right|^\alpha
          \delta\left(\alpha(m')-\Hat v-\Hat\Gamma(E)\right)dE
        \right)\right]^{1/(2-\alpha)}
    \end{split}
  \end{equation*}
  for any $\alpha\in(0,1)$.  Using again the H{\"o}lder inequality we
  get
  \begin{equation*}
    \begin{split}
      &\EE\left(|\alpha(m')-\Hat v|^\alpha\sup_{\|f\|=1}\int_F
        \left|\langle f,\gamma_\Lambda(E)(M_\Lambda(E)-\Hat
          A_\Lambda)^{-1}\delta_{m'}\rangle\right|^\alpha
        \delta\left(\alpha(m')-\Hat v-\Hat\Gamma(E)\right)dE
      \right)\\
    &\leq 2\EE (|\Hat v|^\alpha)^{\alpha/\delta}
    \left[\EE\left(\sup_{\|f\|=1} \int_F \left|\langle
          f,\gamma_\Lambda(E)(M_\Lambda(E)-\Hat
          A_\Lambda)^{-1}\delta_{m'}\rangle\right|^{s}
        \delta\left(\alpha(m')-\Hat v-\Hat\Gamma(E)\right)dE\right)
    \right]^{\alpha/s}
    \end{split}
  \end{equation*}
  with $\alpha/s+\alpha/\delta=1$. Using \eqref{eq-Efin}, we estimate,
  \begin{equation*}
    \begin{split}
      &\EE\left(\sup_{\|f\|=1}\int_F \left|\langle
          f,\gamma_\Lambda(E)(M_\Lambda(E)-\Hat
          A_\Lambda)^{-1}\delta_{m'}\rangle\right|^{s}
        \delta\left(\alpha(m')-\Hat v-\Hat\Gamma(E)\right)dE
      \right)\\
      &\le \int_F \EE\left(\sup_{\|f\|=1}\left| \langle
          f,\gamma_\Lambda(E)(M_\Lambda(E)-\Hat
          A_\Lambda)^{-1}\delta_{m'}\rangle\right|^{s}\right)
      \rho\left(\Hat v+\Hat\Gamma(E)\right)dE\\
      &\le R |F| \sup_{E\in F}
      \EE\left(\left\|\left(\gamma_\Lambda(E)(M_\Lambda(E)-\Hat
            A_\Lambda)^{-1}\delta_{m'}\right)_{m,j}\right\|^s\right)\\
      &\le R |F| C \sup_{E\in F} \EE\left(\left|(M_\Lambda(E)-\Hat
          A_\Lambda)^{-1}(m,m')\right|^s\right)\\&\hskip3cm+ R |F| C
      \sup_{E\in F} \EE\left(\left|(M_\Lambda(E)-\Hat
          A_\Lambda)^{-1}(m+h_j,m')\right|^s\right)\\
    &\le R |F| C \left( A e^{-a|m-m'|} + A e^{-a|m+h_j-m'|} \right)\le
    A R |F| C (1+e^a) e^{-a|m-m'|},
    \end{split}
  \end{equation*}
  where $R=\sup\rho$ and
  \begin{equation*}
    C=\max\Big(\sup_{E\in F}
    \Big\|\dfrac{\varphi_j(\cdot,E)}{\varphi_j(l_j,E)}\Big\|^s,
    \sup_{E\in F}\Big\|\dfrac{\phi_j(\cdot,E)}{\varphi_j(l_j,E)}\Big\|^s
    \Big).
  \end{equation*} 
  Finally, as follows from \eqref{eq-mup2}, one has
  \begin{equation}
    \label{eq-mup3}
    \EE\left(\sup_{\|f\|=\|g\|=1} \int_F \big| \langle
    f,\gamma_\Lambda(E)\Hat \varphi_{E}\rangle \langle
    \gamma_\Lambda(E)\Hat \varphi_{E}, \Hat P(E) P(E)g\rangle\big|
    \delta(\alpha(m')-\Hat v-\Hat\Gamma(E)) dE \right)\le \Hat B
    e^{-\Hat c|m-m'|}
  \end{equation}
  with some $\Hat B, \Hat c>0$.

  One can estimate the second term on the right-hand side of
  \eqref{eq-mup1} in exactly the same way.  Using \eqref{eq-pepe} and
  \eqref{eq-gp} and the inequality \eqref{eq-est2}, after similar
  steps, one gets
  \begin{equation*}
    \begin{split}
      \EE&\left(\sup_{\|f\|=1}\left(\int_F \left|\langle
            f,\gamma_\Lambda(E)(M_\Lambda(E)-\Tilde
            A_\Lambda)^{-1}\delta_{m'+h_{j'}}\rangle\right|^{s}
          \delta\left(\alpha(m'+h_{j'})-\Tilde
            v-\Tilde\Gamma(E)\right)dE\right)\right)\\
      &\le \int_F \EE\left(\sup_{\|f\|=1}\left| \langle
          f,\gamma_\Lambda(E)(M_\Lambda(E)-\Tilde
          A_\Lambda)^{-1}\delta_{m'+h_{j'}}\rangle\right|^{s}\right)
      \rho\left(\Tilde v+\Tilde\Gamma(E)\right)dE\\
      &\le R |F| \sup_{E\in
        F}\EE\left(\left\|\left(\gamma_\Lambda(E)(M_\Lambda(E)-\Tilde
            A_\Lambda)^{-1}\delta_{m'+h_{j'}}\right)_{m,j}\right\|^s\right)\\
      &\le R |F| C \sup_{E\in F} \EE\left(\left|(M_\Lambda(E)-\Tilde
          A_\Lambda)^{-1}(m,m'+h_{j'})\right|^s\right)\\
      &\hskip1cm+ R |F| C \sup_{E\in F}
      \EE\left(\left|(M_\Lambda(E)-\Tilde
          A_\Lambda)^{-1}(m+h_j,m'+h_{j'})\right|^s\right)\\ &\le R
      |F| C \left(A e^{-a|m-m'-h_{j'}|} + A e^{-a|m+h_j-m'-h_{j'}|}
      \right)\\ &\le A R |F| C (e^a+e^{2a}) e^{-a|m-m'|},
    \end{split}
  \end{equation*}
  which gives, for some positive constants $\Tilde B$ and $\Tilde c$,
  \begin{equation}
    \label{eq-mup4}
    \EE\left(\sup_{\|f\|=\|g\|=1} \int_F \left| \langle
        f,\gamma_\Lambda(E)\Tilde \varphi_{E}\rangle \langle
        \gamma_\Lambda(E)\Tilde \varphi_{E}, \Tilde P(E)
        P(E)g\rangle\right|\delta(\alpha(m'+h_{j'})-\Tilde
      v-\Tilde\Gamma(E)) dE\right) \le \Tilde B e^{-\Tilde c|m-m'|}.
  \end{equation}
  Substituting \eqref{eq-mup3} and \eqref{eq-mup4} into
  \eqref{eq-mup1} we obtain the requested inequality \eqref{eq-PF}.
\end{proof}

\section{Finite volume criteria}

We now will show how the results of~\cite{ASFH} apply in our case.

We need some constants characterizing the distribution of the coupling
constants.  Let $s\in(0,1)$. Define 
\begin{equation*}
  C_s=\sup_{A\in\mathcal{M}_{2\times2}(\C)}\int\int \rho(du)\rho(dv)
  \Big|\left[\Big(A -\diag(u,v)\big)^{-1}\right]_{jk}\Big|^s.
\end{equation*}
In~\cite{ASFH}, it is shown that $C_s$ is finite. It is also shown
that for any $s\in(0,1/4)$, if, for $a,b,c\in\CC$, we define
$f(V):=(V-a)^{-1}$, $g(V):=(V-b)(V-c)^{-1}$, then
\begin{equation*}
  D_s=\sup_{a,b,c}\frac{\EE \big(|f(V)|^s |g(V)|^s\big)}
  {\EE \big(|f(V)|^s\big)\EE\big(|g(V)|^s\big)}<+\infty.
\end{equation*}
We set $\Tilde C_s:= C_s D^2_s$.

In the standard basis of $l^2(\ZZ^d)$ the operator $M(E)+a(E)$ is
given by the matrix $\big(\tau_{m,m'}(E)\big)_{m,m'\in\ZZ^d}$ with
\begin{equation}
  \label{eq:1}
  \tau_{m,m'}(E)=\begin{cases}
    0 & m=m',\\
    b_j(E), & m=m'\pm h_j,\\
    0, & |m-m'|>1.
  \end{cases}
\end{equation}
Let $s\in(0,1/4)$. For any $\Lambda\subset\ZZ^d$ denote
\begin{equation*}
  T^s_{m,\partial\Lambda}(E):=\sum_{n\in W} |\tau_{m,n}(E)|^s,\quad
  m\in\ZZ^d,
  \quad
  W=\begin{cases}
    \ZZ^d\setminus\Lambda, & m\in\Lambda,\\
    \Lambda, & m\notin\Lambda.
  \end{cases}
\end{equation*}
Furthermore, set
\begin{equation*}
  \Theta^s_{\Lambda}(E):=\sum_{m\in\Lambda} T^s_{m,\partial\Lambda}(E),
\end{equation*}
and
\begin{gather*}
  k_\Lambda(m,n;E):=|\tau_{m,n}(E)|^s I_1(m,n)+
  T^s_{m,\partial\Lambda}(E) T^s_{n,\partial\Lambda}(E) \dfrac{\Tilde
    C_s}{\lambda^s} I_2(m,n)\\{}
  +T^s_{m,\partial\Lambda}(E)T^s_{n,\partial\Lambda}(E)
  \Big(\dfrac{\Tilde C_s}{\lambda^s}\Big)^2\Theta^s_\Lambda(E)I_3(u,v),\\
  \intertext{where} I_1(m,n)=\begin{cases}
    1 & m\in\Lambda, n\notin\Lambda,\\
    0, & \text{otherwise,}
  \end{cases},\ 
  I_2(m,n)=\begin{cases}
    1 & m\in\Lambda,\\
    0, & \text{otherwise,}
  \end{cases},\ 
  I_3(m,n)=\begin{cases}
    1 & m\in\Lambda, n\in\Lambda,\\
    0, & \text{otherwise.}
  \end{cases}
\end{gather*}

Theorem~3.2 in \cite{ASFH} and the remark thereafter read in our case as
follows.

\begin{prop}
  \label{prop-fin}
  Take any interval $X\subset\RR$ free of Dirichlet eigenvalues.
  Assume that there exist $\beta\in (0,1)$ and $s\in(0,1/4)$ such that
  for all $E\in X$ there exists a finite $\Lambda\subset\ZZ^d$ with
  $0\in\Lambda$ obeying
  \begin{equation}
    \label{eq-klam}
    \sup_{W\subset\Lambda} \sum_{(m,n)\in
      \Lambda\times(\ZZ^d\setminus\Lambda)}
    \EE\Big(\big|
    \big(M_W(E)-\lambda A_{W,\omega}\big)^{-1}(0,m)
    \big|^s\Big) k_\Lambda(m,n;E)\le \beta.
  \end{equation} 
  Then, there exist $B,c>0$ such that, for any finite
  $\Theta\subset\ZZ^d$, any $m\in\Theta$, and any $E\in X$, one has
  \begin{equation*}
    \sum_{m'\in\Theta} \EE\Big(\big|
    \big(M_\Theta(E)-\lambda A_{\Theta,\omega}\big)^{-1}(m,m')
    \big|^s\Big) e^{c|m-m'|}\le B.
  \end{equation*}
\end{prop}
We note that the possibility to choose the constant $B$ independent of
$E$ follows from Eq.~(3.20) in \cite{ASFH}.

It is also important to emphasize that in the sum \eqref{eq-klam} the
coefficients $k_\Lambda(m,n;E)$ are non-zero only if simultaneously
$\dist(n,\Lambda)=1$ and $\dist(m,\ZZ^d\setminus\Lambda)=1$.

For convenience, we formulate proposition~\ref{prop-fin} for the
special case $\Lambda=\{0\}$, which will be used below.

\begin{prop} 
  \label{prop-fin2} 
  Take any $X\subset\RR$ free of Dirichlet eigenvalues.  Assume that
  there exists $\beta\in(0,1)$ and $s\in(0,1/4)$ such that for all
  $E\in X$ one has
  \begin{equation}
          \label{eq-abc}
    c(E) \Big( 1+c(E)\,\dfrac{\Tilde C_s}{\lambda^s}\Big) \int_{\alpha_-}^{\alpha_+}
    \dfrac{1}{\big|\, a(E)+\lambda V\big|^s}\,\rho(dV)<\beta, \quad
    c(E):=2\sum_{j=1}^d\big|b_j(E)\big|^s.
  \end{equation}
  Then there exist $B,c>0$ such that for any finite
  $\Lambda\subset\ZZ^d$, for any $m,m'\in\Lambda$, and any $E\in X$
  there holds
  \begin{equation*}
    \EE\left(\Big|
    \big(M_\Lambda(E)-\lambda A_{\Lambda,\omega}\big)^{-1}(m,m')
    \Big|^s\right) \le B e^{-c|m-m'|}.
  \end{equation*}
\end{prop}

The condition $s\in(0,1/4)$ is needed for the so-called decoupling
property to hold (see~\cite{Az}). Actually a revision of the proofs
in~\cite{ASFH} shows that the decoupling property is not necessary in
our case as the operators $M(E)$ do not depend on the random
variables, and one can obtain some finite volume criteria with any
power $s\in(0,1)$. 

The following theorem summarizes all the above localization conditions
for quantum graphs.

\begin{theorem}
  \label{thm-loc}
  Let $X\subset\RR$ be free of the Dirichlet eigenvalues and have a
  finite Lebesgue measure. Assume that the assumptions of
  proposition~\ref{prop-fin} are satisfied, then $H_{\lambda,\omega}$
  has only pure point spectrum in $X$.
\end{theorem}

\begin{proof}
  By proposition~\ref{prop-fin}, there exist $B,c>0$ such that for all
  finite $\Lambda\subset\ZZ^d$ and all $E\in X$ one has $\EE\Big|
  \big(M_\Lambda(E)-\lambda A_{\Lambda,\omega}\big)^{-1}(m,m') \Big|^s
  \le B e^{-c|m-m'|}$. Then, by proposition~\ref{prop-measure}, one
  has $\EE \Big(\mu^{(m,j),(m',j')}(X)\Big)\le B e^{-c|m-m'|}$,
  $B,c>0$.  Hence, for any $(m,j)$ the following bound holds
  \begin{equation*}
    \EE\Big(
    \sum_{m'\in\ZZ^d}\sum_{j'=1}^d
    \mu^{(m,j),(m',j')}(X) \Big) \le Bd \sum_{m'\in\ZZ^d} e^{-c|m'|}<\infty,
  \end{equation*}
  and the spectrum of $H_{\lambda,\omega}$ in $X$ is pure point by
  corollary~\ref{corol1}.
\end{proof}

\section{Strong disorder localization}\label{ss-strong}

Here we are going to exhibit assumptions ensuring that one obtains
dense pure point spectrum in some regions for sufficiently large
constant $\lambda$. To garantee the presence of the dense pure point
spectrum, it is necessary to show the overlapping of the spectrum of
$H_{\lambda,\omega}$ with the region where the assumptions of
proposition~\ref{prop-fin} are fulfilled.

\begin{prop}\label{prop-main}
  For any $E_0\in\RR$ and any $\varepsilon>0$ there exists
  $\lambda_0>0$ such that the spectrum of $H_{\lambda,\omega}$ lying
  in $(-\infty,E_0)$ but outside the $\varepsilon$-neighborhoods of
  the Dirichlet eigenvalues is pure point for all $\lambda>\lambda_0$.
\end{prop}

\begin{proof}
  We use the single point criterium, proposition~\ref{prop-fin2}.
  Denote by $X$ the half-axis $(-\infty,E_0)$ without he
  $\varepsilon$-neighborhoods of the Dirichlet eigenvalues.  Due to
  the asymptotics \eqref{eq-asymp}, one can estimate, for some
  $\delta>0$, $|\varphi_j(l_j;E)|\ge \delta>0$ uniformly for $E\in X$.
  Hence for $E\in X$ one has $|b_j(E)|\le B$, $|c(E)|\le B$ for some
  $B>0$, and, moreover, due to~\eqref{eq-asymp},
  $b_j(E)=O(e^{-\alpha\sqrt{-E}})$, $c(E)=O(e^{-s\alpha\sqrt{-E}})$
  for some $\alpha>0$ as $E\to-\infty$.

  Pick $s\in(0,1/4)$.  As the density $\rho$ is bounded, say, $\rho\le
  R$, one has
  \begin{equation*}
    \begin{split}
      \int_{\alpha_-}^{\alpha_+} \big|\, a(E)+\lambda
      V\big|^{-s}\,\rho(dV) &\le R \int_{\alpha_-}^{\alpha_+} \big|\,
      a(E)+\lambda V\big|^{-s}dV\le\frac
      R\lambda\int_{\alpha_-/\lambda-a(E)} ^{\alpha_+/\lambda-a(E)}
      \big| V\big|^{-s} dV\\&\leq \frac{2R}{\lambda s}
      \left|\frac{\alpha_+-\alpha_-}\lambda\right|^{1-s}
      \leq\frac{C}{\lambda^s}.
    \end{split}
  \end{equation*}
  Therefore,
  \begin{equation}
    \label{eq-celoc}
    c(E) \Big( 1+c(E)\,\dfrac{\Tilde C_s}{\lambda^s}\Big) \int
    \dfrac{1}{\Big|\,a(E)+\lambda V\Big|^s}\,\rho(dV)
    \le \tilde C(E)(\lambda^{-s}+\lambda^{-2s})
  \end{equation}
  where $\tilde C(E)$ is bounded in $X$. Hence, the left-hand side
  of~\eqref{eq-celoc} tends to $0$ uniformly in $X$ as $\lambda$
  becomes large.  The spectrum if $H_{\lambda,\omega}$ in any compact
  subset of $X$ is then pure point by theorem~\ref{th-loc}.
\end{proof}

Proposition~\ref{prop-main} does not guarantee that there is some
spectrum in the set considered. To show the presence of a dense point
spectrum we use the estimates of subsection~\ref{ss-ran} to obtain

\begin{theorem}\label{th1}
  Assume that $\alpha_-<0$. Then, for any $\varepsilon>0$, there
  exists $\lambda_0>0$ such that the spectrum of $H_{\lambda,\omega}$
  in $(-\infty,\inf\spec H_0-\varepsilon)$ is dense pure point for
  $\lambda>\lambda_0$.
\end{theorem}

\begin{theorem}\label{th2}
  Let $0\in[\alpha_-,\alpha_+]$. Then, for any $E_0>\inf\spec H^0$ and
  any $\varepsilon>0$, there exists $\lambda_0>0$ such that the
  spectrum of $H_{\lambda,\omega}$ lying in $(-\infty,E_0)$ but
  outside the $\varepsilon$-neighborhoods of the Dirichlet eigenvalues
  is dense pure point for all $\lambda>\lambda_0$.
\end{theorem}

Both theorems~\ref{th1} and \ref{th2} are direct consequences of
proposition~\ref{prop-main}.  The discussion of
subsection~\ref{ss-ran} shows that the intersection of the spectrum of
$H_{\lambda,\omega}$ with the sets considered is non-empty for large
$\lambda$.

In theorems~\ref{th1} and~\ref{th2}, we only stated the localized
spectrum. Clearly by virtue of remark \ref{rem3}, we get also exponential decay of the
eigenfunctions and dynamical localization.

We were not able to study the effect of the strong disorder in neighborhoods on the Dirichlet eigenvalues.
The reason is that in these neighborhoods the expression $a(E)$ in the single point criterion, proposition~\ref{prop-fin2},
becomes unbounded, and hence even large $\lambda$ gives no possibility to control the value of the integral
in \eqref{eq-abc}. Moreover, if both the constants $\alpha_-$ and $\alpha_+$ are positive, then, by discussion
of subsection~\ref{ss-ran}, the whole spectrum is concentrated in these neighborhoods, so the above theorems
do not provide any localilization result in this case. In the next section we will be
able to fill this gap at least partially and to prove localization near the spectral edges independently
of their location.

\section{Localization at band edges}\label{ss-lif}

Here we are going to show the presence of the dense pure point
spectrum at the edges of the spectrum of $H_\omega$.

The starting point will the following simple observation.
\begin{prop} \label{prop-bd} Let $E_0\in\spec H_\omega\setminus H^0$.
  If for some $\varepsilon>0$ one has
  $(E_0-\varepsilon,E_0)\notin\spec H_\omega$ or
  $(E_0,E_0+\varepsilon)\notin\spec H_\omega$, then either
  $\inf\spec\big( M(E_0)-A_\omega\big)=0$ or $\sup\spec\big(
  M(E_0)-A_\omega\big)=0$.  In other words, if $E_0\notin\spec H^0$ is
  at the border of the spectrum of $H_\omega$, then $0$ is a border of
  the spectrum of $M(E_0)-A_\omega$.
\end{prop}

\begin{proof}
  As~\eqref{eq-mea} shows, the spectrum of $M(E)-A_\omega$ is a
  segment $[m_-(E),m_+(E)]$ whose ends $m_-(E):=\inf \Sigma_M(E)$ and
  $m_+(E):=\sup\Sigma_M(E)$ depend continuously on $E$.  As
  $E_0\in\spec H_\omega$, one has necessarily $0\in\Sigma_M(E)$, i.e.
  $m_-(E_0)m_+(E_0)\le 0$.  If one had $m_-(E_0)m_+(E_0)<0$, i.e.
  $m_-(E_0)<0$ and $m_-(E_0)>0$, then the inequality $m_-(E)m_+(E)<0$
  would hold also for $E\in(E_0-\varepsilon,E_0+\varepsilon)$ with
  some $\varepsilon>0$. But this would mean that
  $(E_0-\varepsilon,E_0+\varepsilon)\subset\spec H_\omega$, which
  contradicts the assumptions.  Therefore, the only possibility is
  $m_-(E_0)\cdot m_+(E_0)=0$.
\end{proof}

\begin{theorem}
  Let $E_0\notin\spec H^0$ be at the border of the spectrum of
  $H_\omega$.  Then the spectrum of $H_\omega$ in some neighborhood of
  $E_0$ is pure point almost surely.
\end{theorem}
In the present case, remark~\ref{rem3} gives also exponential decay of
the eigenfunctions and dynamical localization.

\begin{proof}
  Proposition~\ref{prop-bd} shows that $0$ is an edge of the spectrum
  of $M(E_0)-A_\omega$.  To be definite, we consider only the case
  $\inf \spec\big(M(E_0)-A_\omega\big)=0$; the other case can be
  studied in the same way. Note that due to the variational principle
  one has $M_W(E_0)-A_{W,\omega}\ge 0$ for any $W\subset\ZZ^d$.

  Let us do first some preparations.  For any $W\subset\ZZ^d$ and
  $\varepsilon>0$ consider the following subset of $\Omega$:
  \begin{equation*}
    \Omega(\varepsilon,W):=\big\{\omega\in\Omega:\,
    \inf \spec \big(M_W(E_0)-A_{W,\omega}\big)\le\varepsilon
    \big\}.
  \end{equation*}
  Clearly, by the variational principle one has
  $\Omega(\varepsilon,W)\subset \Omega(\varepsilon,W')$ if $W\subset
  W'$.

  Let $\cN(\lambda)$ by the integrated density of states corresponding
  to $M(E_0)-A_\omega$.  Denote $\Lambda_N:=\{m\in\ZZ^d:\,\max_j
  |m_j|\le N\}$, $N\in\NN$.  It is known~\cite{KW} that with some
  $C>0$ one has
  \begin{equation*}
    \PP \big(\Omega(\varepsilon,\Lambda_N)\big)\le C N^d \cN(\varepsilon)
    \quad\text{for any } N\ge 1.
  \end{equation*}
  At the same time, one has the Lifshitz asymptotics for
  $\cN(\varepsilon)$, i.e. there exists $\varepsilon_0>0$ and $\eta>0$
  such that $\cN(\varepsilon)\le e^{-\varepsilon^{-\eta}}$,
  $\varepsilon\in(0,\varepsilon_0)$. Indeed, by~\eqref{eq:1}, the
  Fourier symbol of $M(E_0)$ is of the form $\D\sum_{j=1}^d b_j\cos \theta_j
  -a$ with $b_j\ne 0$, hence, one can apply to $M(E_0)+A_\omega$  the techniques of
  \cite{FK1} to obtain that $\D\log|\log\cN(\varepsilon)|=-\frac
  d2\log\varepsilon(1+o(1))$ when $\varepsilon\to0+$.

  For any finite $W\subset\ZZ^d$ and $\varepsilon>0$ denote
  \begin{equation*}
    \Tilde\Omega(\varepsilon,W):= \big\{\omega\in\Omega:\,
    \inf\spec\big(M_W(E)-A_{W,\omega}\big)\le \varepsilon\text{ for
      some } E,\, |E-E_0|<\varepsilon \big\}.
  \end{equation*}
  Note that the condition $\inf\spec \big(M_W(E)-A_{W,\omega}\big)\le
  \varepsilon$ is equivalent to the existence of a non-zero $\xi_E \in
  l^2(W)$ with
  \begin{equation}
    \label{eq-mwe}
    \big\langle\xi_E,\big(M_W(E)-A_{W,\omega}\big)\xi_E\big\rangle\le
    \varepsilon\cdot \|\xi_E\|^2.
  \end{equation}
  Representing $M(E)=M(E_0)+(E-E_0)B(E)$, where $\|B(E)\|\le D$ for
  some $D>0$ in a neighborhood of $E_0$, one immediately sees that
  \eqref{eq-mwe} implies
  $\big\langle\xi_E,\big(M_W(E_0)-A_{W,\omega}\big)\xi_E\big\rangle\le
  (D+1)\varepsilon\|\xi_E\|^2$, which means
  $\inf\spec\big(M_W(E_0)-A_{W,\omega}\big)\le (D+1)\varepsilon$.
  This shows the inclusion $\Tilde\Omega(\varepsilon,W)\subset
  \Omega((D+1)\varepsilon,W)$.

  With the above preparations we just need to repeat the basic steps
  from~\cite[Section~2]{FK}.  It is sufficient to show that there
  exists a neighborhood $X$ of $E_0$ where the assumptions of
  proposition~\ref{prop-fin} are satisfied for $\Lambda=\Lambda_N$
  with a suitable $N$.

  Let us fix some $s\in(0,1/4)$.  Consider any $W\subset\Lambda_N$. As
  shown above, one has
  $\Tilde\Omega(\varepsilon,W)\subset\Omega((D+1)\varepsilon,
  W)\subset \Omega((D+1)\varepsilon,\Lambda_N)$.  Subsequently, for
  $\varepsilon\in (0,\varepsilon')$ with some $\varepsilon'>0$, one
  has
  \begin{equation}
    \label{eq-PP1}
    \PP\big(\Tilde\Omega(\varepsilon,W)\big)\le
    \PP\big(\Omega((D+1)\varepsilon,\Lambda_N)\big)\le 
    C N^d e^{-\varepsilon^{-\eta}},\quad \eta>0.
  \end{equation}
  For $\omega\notin \Tilde\Omega(\varepsilon,W)$ one can use the
  Combes-Thomas estimates, see e.g. \cite[lemma 6.1]{FK}, which gives
  that for some $C', r>0$ one has
  \begin{equation}
    \label{eq-combes}
    \big|\big(M_W(E)-A_{W,\omega}\big)^{-1}(m,m')\big|\le C' e^{-r|m-m'|}.
  \end{equation}
  Eq.~(6.1) in~\cite{FK} shows that the constants $C'$ and $r$ can be
  chosen independent of $W$ as in our case $\inf\spec
  \big(M_W(E)-A_{W,\omega})>\varepsilon$.

  Take any $s'\in(s,1)$, then for any $E$ with $|E-E_0|<\varepsilon$
  one has also an a priori estimate
  \begin{equation}
    \label{eq-emw}
    \EE\Big(\big|\big(M_W(E)-A_{W,\omega}\big)^{-1}(m,m')\big|^{s'}\Big)\le C_{s'},
  \end{equation}
  see \cite[lemma 2.1]{ASFH}.

  Now we have
  \begin{equation*}
    \begin{split}
      \EE\Big(\big|\big(M_W(E)-A_{W,\omega}\big)^{-1}(m,m')\big|^s\Big)
      &=\EE\Big(\big|\big(M_W(E)-A_{W,\omega}\big)^{-1}(m,m')\big|^s
      \mathbf{1}_{\omega\in\Tilde\Omega(\varepsilon,W)}\Big)\\
      &\hskip2cm+\EE\Big(\big|\big(M_W(E)-A_{W,\omega}\big)^{-1}(m,m')\big|^s
      \mathbf{1}_{\omega\notin\Tilde\Omega(\varepsilon,W)}\Big).
    \end{split}
  \end{equation*}
  Using~\eqref{eq-combes} we obtain easily
  \begin{equation*}
    \EE\Big(\big|\big(M_W(E)-A_{W,\omega}\big)^{-1}(m,m')\big|^s
    \mathbf{1}_{\omega\notin\Tilde\Omega(\varepsilon,W)}\Big) 
    \le B e^{-b|m-m'|},\quad B,b>0.
  \end{equation*}
  Using the H{\"o}lder inequality, \eqref{eq-PP1} and \eqref{eq-emw},
  for some $C'>0$ and $\gamma>0$, one has
  \begin{equation*}
    \begin{split}
      \EE\Big(\big|\big(M_W(E)-A_{W,\omega}\big)^{-1}(m,m')\big|^s
      \mathbf{1}_{\omega\notin\Tilde\Omega(\varepsilon,W)}\Big)
      &\le\Big(\EE\big|\big(M_W(E)-A_{W,\omega}\big)^{-1}(m,m')\big|^s\Big)^{s/s'}
      \PP(\Tilde\Omega(\varepsilon,W))^{(s'-s)/s}\\
      &\le C' N^d e^{-\varepsilon^{-\gamma}}.
    \end{split}
  \end{equation*}
  Finally,
  \begin{equation}
    \label{eq-EE}
    \EE\big|\big(M_W(E)-A_{W,\omega}\big)^{-1}(m,m')\big|^s\le  B
    e^{-b|m-m'|} + C' N^d e^{-\varepsilon^{-\gamma}}.
  \end{equation}
  Now let us estimate the sum \eqref{eq-klam}. We emphasize again that
  the coefficients $k_{\Lambda_N}(m,n;E)$ in this sum are non-zero
  only if simultaneously $\dist(n,\Lambda_N)=1$ and
  $\dist(m,\ZZ^d\setminus\Lambda_N)=1$. Moreover, the non-zero terms
  are uniformly bounded in a neighborhood of $E_0$,
  $k_{\Lambda_N}(m,n;E)\le K$, $K>0$.  Therefore, using \eqref{eq-EE},
  \begin{equation*}
    \begin{split}
      \sum_{\substack{m\in\Lambda_N\\n\in\ZZ^d\setminus\Lambda_N}}
      \EE\Big(\big| \big(M_W(E)-\lambda A_{W,\omega}\big)^{-1}(0,m)
      \big|^s\Big) k_\Lambda(m,n;E)&\le K
      \sum_{\substack{m\in\Lambda_N\\\dist(m,\ZZ^d\setminus\Lambda_N)=1\\
          n\notin\Lambda_N\\\dist(n,\Lambda_N)=1}}
      \Big( B e^{-b|m|} + C' N^d e^{-\varepsilon^{-\gamma}}\Big)\\
      &\le K' N^{2d} \Big( B e^{-b N} + C' N^d
      e^{-\varepsilon^{-\gamma}}\Big).
    \end{split}
  \end{equation*}
  Now choosing, for example, $N\sim \varepsilon^{-1}$ one can make the
  sum as small as needed for sufficiently small $\varepsilon$.  The
  spectrum of $H_\omega$ near $E_0$ is then pure point by
  theorem~\ref{th-loc}.
\end{proof}

\appendix

\section{Proofs of propositions~\ref{prop-conv}   and~\ref{prop-comp}}
\label{sec:appendix}

In this subsection, we prove some auxiliary results on the finite
volume approximation for $H_A$ defined in
section~\ref{sec:local-cond-quant}.

\begin{proof}[\bf Proof of proposition~\ref{prop-conv}]
  To prove the convergence, we will use the following variant of
  theorem VIII.1.5 from \cite{Kato}: \emph{Let $T_n$, $T$ be
    self-adjoint operators. Assume that there exists a domain $D$ of
    essential self-adjointness (or a core) for $T$ such that every
    function $f$ from $D$ belongs to $\dom T_m$ for $m$ sufficiently
    large and $T_n f\to Tf$ for any such $f$. Assume that, for at
    least one non-real $z$, the sequence $\|(T_n-z)^{-1}\|$ is
    bounded, then $T_n$ converges to $T$ in the strong resolvent
    sense.}

  In our case, take as $D$ the set of the functions $f\in \dom H_A$
  having a compact support.  Clearly, any such $f$ lies in $\dom
  H^{\Lambda_n}_A$ for $n$ sufficiently large, and
  $H^{\Lambda_n}_\alpha f$ just coincides with $H_A f$ for such $n$.
  Let us show that $D$ is a domain of essential self-adjointness for
  $H_A$.

  Choose functions $u_j\in C^\infty[0,l_j]$ such that $u_j$ is $1$ in
  a neighborhood of $0$ and is $0$ in a neighborhood of $l_j$. Take an
  arbitrary $f\in\dom H_A$.  For $M\in\NN$ denote $f^M:=(f^M_{m,j})$
  with
  \begin{equation*}
    f^M_{m,j}(t):=\begin{cases}
      f_{m,j}(t), & m\in\Lambda_{M-1},\\
      u_j(t) f_{m,j}(t), & (m,j)=m\to m',\,   m\in\Lambda_M\setminus \Lambda_{M-1}, m'\notin \Lambda_M,\\ 
      u_j(l_j-t) f_{m,j}(t), & (m,j)=m\to m',\,
      m'\in\Lambda_M\setminus \Lambda_{M-1}, m\notin \Lambda_M,\\ 
      0, & \text{otherwise}.
    \end{cases}
  \end{equation*}
  Clear, $f^M\in D$ and $\D f^M\vers_{M\to+\infty}f$.  Note that $H_A
  (f^M-f)=(F^M_{m,j})$ with
  \begin{align*}
    F^M_{m,j}&= u''_j f_{m,j}+2u'_j f'_{m,j}+ \big(u_j-1\big)\big(-f''_{m,j}+U_jf_{m,j}\big),\\
    &\qquad (m,j)=m\to m', \quad m\in\Lambda_M\setminus \Lambda_{M-1}, m'\notin \Lambda_M,\\
    F^M_{m,j}&=u''_j(l_j-\cdot) f_{m,j}-2u'_j(l_j-\cdot) f'_{m,j}+ \big(u_j(l_j-\cdot)-1\big)\big(-f''_{m,j}+U_jf_{m,j}\big),\\
    &\qquad (m,j)=m\to m', \quad m'\in\Lambda_M\setminus
    \Lambda_{M-1}, m\notin \Lambda_M,
  \end{align*}
  and all other components $F^M_{m,j}$ equal to $0$. As $(f_{m,j})\in
  \HH$, $(f'_{m,j})\in \HH$ and $(-f''_{m,j}+U_j f_{m,j})\in \HH$, one
  has $\D H_A(f^M-f)\vers_{M\to\infty}0$. This shows that $H_A$ is
  essentially self-adjoint on $D$.

  To conclude the proof of the strong resolvent convergence it remains
  to show that that the norms $\|(H^{\Lambda_n}_\alpha-E)^{-1}\|$ are
  uniformly bounded for at least one non-real $E$.  Take an arbitrary
  $E$ with $\Im E\ne 0$. By \eqref{eq-me} one has $\Im M(E)/\Im E\ge
  c$ for some $c>0$.  In particular, for any $n$ one has
  \begin{equation*}
    \bigg|\Big\langle \big(\Pi_{\Lambda_n} (M(E)-A)
    \Pi_{\Lambda_n}\big)\Pi_{\Lambda_n}\xi,
    \Pi_{\Lambda_n}\xi\Big\rangle\bigg|
    = \Big|\big\langle (M_{\Lambda_n}(E)-A_{\Lambda_n})
    \Pi_{\Lambda_n}\xi,\Pi_{\Lambda_n}\xi\big\rangle\Big|\ge
    c\|\Pi_{\Lambda_n}\xi\|^2,
  \end{equation*}
  which means that $M_{\Lambda_n}(E)-A_{\Lambda_n}$ has a bounded
  inverse, and that
  $\big\|\big(M_{\Lambda_n}(E)-A_{\Lambda_n}\big)^{-1}\big\|\le
  c^{-1}$.  Now it follows from \eqref{eq-krein2} that the norms
  $\|(H^{\Lambda_n}_\alpha-E)^{-1}\|$ are uniformly bounded for any
  non-real $E$.
\end{proof}

\begin{proof}[\bf Proof of proposition~\ref{prop-comp}]
We note first that $T$ is an
  integral operator whose integral kernel is
  \begin{equation*}
    T\big((m,j,t),(m',j',t')\big)=\begin{cases}
      G_A\big((m,j,t),(m',j',t')\big), & m\in\Lambda,\\
      0, & \text{otherwise}.
    \end{cases}
  \end{equation*}
  Hence
  \begin{equation*}
    \|T\|^2_{HS}=\sum_{m\in\Lambda}\sum_{m'\in\ZZ^d}\sum_{j,j'=1}^d
    \int_{0}^{l_j}\int_0^{l_{j'}}
    \big|G_A\big((m,j,t),(m',j',t')\big)\big|^2\,dt'\,dt.
  \end{equation*}
  Using the explicit form \eqref{eq-green} for $G_A$ and the bounds
  $\|\varphi_j\|,\|\phi_j\|\le C$ (with $C$ independent of $j$) we
  obtain
  \begin{equation}
    \label{eq-SHS}
    \|T\|^2_{HS}\le 2\Big(\sum_{j= 1}^d\|G_j\|^2_{HS}\Big)|\Lambda|
    + C'\sum_{m\in\Tilde\Lambda}\sum_{m'\in\ZZ^d}\big|(M(E)-A)^{-1}(m,m')\big|^2
  \end{equation}
  with some $C'>0$, where $\Tilde\Lambda:=\{m\in\ZZ^d:
  \inf_{m'\in\Lambda}|m-m'|^2\le 2\}$.  Clearly, due to
  \eqref{eq-grf1} the Hilbert-Schmidt norms of $G_j$ are finite.
  Furthermore, as $(M(E)-A)^{-1}$ is bounded for non-real $E$, one has
  \begin{equation*}
    \sum_{m'\in\ZZ^d}\big|(M(E)-A)^{-1}(m,m')\big|^2<\infty
  \end{equation*}
  for any $m\in\ZZ^d$ by the Riesz theorem. Hence, due to the
  finiteness of $\Lambda$ (and of $\Tilde \Lambda$), the sum on the
  right-hand side of \eqref{eq-SHS} is finite.
\end{proof}

\section*{Acknowledgements}

KP was supported by the research fellowship of the Deutsche
Forschungsgemeinschaft (PA 1555/1-1) and by the joint German-New
Zealand project NZL 05/001 funded by the International Bureau at the
German Aerospace Center.


\begin{thebibliography}{99}

\bibitem{Az} M.~Aizenman: \emph{Localization at weak
    disorder\textup{:} some elementary bounds.}  Rev. Math. Phys.
  Special issue (1994) 1163--1182.

\bibitem{ASFH} M.~Aizenman, J.~H.~Schenker, R.~M.~Friedrich,
  D.~Hundertmark: \emph{Finite-volume fractional-moment criteria for
    Anderson localization.}  Commun. Math. Phys. {\bf 224} (2001)
  219--253.

\bibitem{ASW} M. Aizenman, R. Sims, S. Warzel: \emph{Absolutely
    continuous spectra of quantum tree graphs with weak disorder.}
  Commun. Math. Phys. {\bf 264} (2006) 371--389.


\bibitem{AG} W.~O.~Amrein, V.~Georgescu: \emph{On the characterization
    of bound states and scattering states in quantum mechanics.} Helv.
  Phys. Acta {\bf 46} (1973) 635--658.


\bibitem{BMG} A.~Boutet de Monvel, V.~Grinshpun: \emph{Exponential
    localization for multidimensional Schr{\"o}dinger operators with
    random point potentials.} Rev. Math. Phys. {\bf 9} (1997)
  425--451.

\bibitem{BGP1} J.~Br{\"u}ning, V.~Geyler, K.~Pankrashkin: \emph{Spectra
    of self-adjoint extensions and applications to solvable
    Schr{\"o}dinger operators.} Rev. Math. Phys. (to appear), Preprint
  arXiv:math-ph/0611088.

\bibitem{BGP2} J.~Br{\"u}ning, V.~Geyler, K.~Pankrashkin: \emph{Cantor
    and band spectra for periodic quantum graphs with magnetic
    fields.}  Commun. Math. Phys. {\bf 269} (2007) 87--105.


\bibitem{CL} R.~Carmona, J.~Lacroix: \emph{Spectral theory of random
    Schr{\"o}dinger operators.}  (Birkh{\"a}user, Boston, 1990).

\bibitem{DMP} T.~C.~Dorlas, N.~Macris, J.~V.~Pul{\'e}:
  \emph{Characterization of the spectrum of the Landau Hamiltonian
    with delta impurities.}  Commun. Math. Phys. {\bf 204} (1999)
  367--396.

\bibitem{E95} P. Exner: \emph{Lattice Kronig-Penney models.}  Phys.
  Rev. Lett. {\bf 74} (1995) 3503--3506.

\bibitem{EHS} P. Exner, M. Helm, P. Stollmann: \emph{Localization on a
    quantum graph with a random potential on the edges.} Rev. Math.
  Phys. (to appear), Preprint arXiv:math-ph/0612087.


\bibitem{GM} V.~A.~Geyler, V.~A.~Margulis: \emph{Anderson localization
    in the nondiscrete Maryland model.}  Theor. Math. Phys. {\bf70}
  (1987) 133--140.

\bibitem{GS} S.~Gnutzmann, U.~Smilansky: \emph{Quantum
    graphs\textup{:} Applications to quantum chaos and universal
    spectral statistics.} Adv. Phys.  {\bf 55} (2006) 527-–625.

\bibitem{GLV} M.~J.~Gruber, D.~Lenz, I.~Veseli\'c: \emph{Uniform
    existence of the integrated density of states for random
    Schr{\"o}dinger operators on metric graphs over $\ZZ^d$.} Preprint
  arXiv:math.SP/0612743.

\bibitem{HV} M.~Helm, I.~Veseli\'c: \emph{A linear Wegner estimate for
    alloy type Schr{\"o}dinger operators on metric graphs.} Preprint
  arXiv:math.SP/0611609.

\bibitem{HKK} P.~D.~Hislop, W.~Kirsch, M.~Krishna: \emph{Spectral and
    dynamical properties of random models with non-local and singular
    interactions.}  Math. Nachr. {\bf 278} (2005) 627--664.

\bibitem{HP} P.~D.~Hislop, O.~Post: \emph{Anderson localization for
    radial tree-like random quantum graphs.}  Preprint
  arXiv:math-ph/0611022.

\bibitem{Kato} T.~Kato: {\it Perturbation theory for linear operators}
  (Springer-Verlag, Berlin etc., 1966).

\bibitem{Kl:93} F.~Klopp: {\it Localization for semi-classical
    continuous random Schr{\"o}dinger operators II: the random
    displacement model} Helv. Phys. Acta, {\bf 66} (1993) 810--841.

\bibitem{Kl:95a} F.~Klopp: {\it Localisation pour des op{\'e}rateurs de
    Schr{\"o}dinger al{\'e}atoires dans $L^2({{\mathbf {R}}}^d)$: un
    mod{\`e}le semi-classique} Ann.~Inst.~Fourier {\bf 45} (1995)
  265--316.

\bibitem{FK1} F.~Klopp: \emph{Band edge behaviour for the integrated
    density of states of random Jacobi matrices in dimension 1.}
  J.~Stat. Phys. {\bf 90} (1998) 927--947.

\bibitem{FK} F.~Klopp: \emph{Weak disorder localization and Lifshitz
    tails.}  Commun. Math. Phys. {\bf 232} (2002) 125--155.

\bibitem{KW} F.~Klopp, T.~Wolff: \emph{Lifschitz tails for
    2-dimensional random Schr{\"o}dinger operators.}  J. Anal. Math.
  {\bf 88} (2002) 63--147.
 
\bibitem{KSch} V. Kostrykin, R. Schrader: \emph{A random necklace
    model.}  Waves Random Media {\bf 14} (2004) S75--S90.

\bibitem{Ku1} P.~Kuchment: \emph{Quantum graphs} I. \emph{Some basic
    structures.}  Waves Random Media {\bf14} (2004) S107--S128.

\bibitem{Ku2} P.~Kuchment: \emph{Quantum graphs} II. \emph{Some
    spectral properties of quantum and combinatorial graphs.} J. Phys.
  A: Math. Gen. {\bf 38} (2005) 4887--4900.


\bibitem{KS} H.~Kunz, B.~Souillard: \emph{Sur le spectre des
    op{\'e}rateurs aux diff{\'e}rences finies al{\'e}atoires.}  Commun.
  Math. Phys. {\bf 78} (1980) 201--246.

\bibitem{LS} B.~M.~Levitan, I.~S.~Sargsyan: \emph{Sturm-Liouville and
    Dirac operators} (Kluwer, Dordrecht etc., 1990).

\bibitem{KP2} K.~Pankrashkin: \emph{Localization effects in a periodic
    quantum graph with magnetic field and spin-orbit interaction.}
  J.~Math. Phys. {\bf 47} (2006) 112105.

\bibitem{KP} K. Pankrashkin: \emph{Spectra of Schr{\"o}dinger operators
    on equilateral quantum graphs.}  Lett. Math. Phys. {\bf 77} (2006)
  139--154.

\bibitem{PF} L.~Pastur, A.~Figotin: \emph{ Spectra of random and
    almost-periodic operators} (Springer, Berlin etc., 1992).

\bibitem{pos} A.~Posilicano: \emph{Self-adjoint extensions of
    restrictions.}  Preprint arXiv:math-ph/0703078.

\bibitem{PS} P.~Stollmann: \emph{Caught by disorder. Bound states in
    random media} (Birkh{\"a}user, Boston, 2001).

\end{thebibliography}
\end{document}